\documentclass[11pt,a4paper]{article}
\usepackage[sc,noBBpl]{mathpazo}
\usepackage[mathscr]{eucal}
\usepackage{amsmath,amsfonts,amssymb,amsthm}
\usepackage{graphicx}
\usepackage{mathtools}
\usepackage{tgpagella}
\usepackage[usenames]{color}
\theoremstyle{plain}
\newtheorem{theorem}{Theorem}
\newtheorem{lemma}[theorem]{Lemma}

\newtheoremstyle{note}{\topsep}{\topsep}{\slshape}{}{\scshape}{}{ }{}
\theoremstyle{note}

\newtheorem{remark}[theorem]{Remark}
\newtheoremstyle{warning}{\topsep}{\topsep}{\color{Blue}\small\slshape}{}{\color{Red}\scshape}{.}{ }{}
\theoremstyle{warning}

\numberwithin{equation}{section}
\numberwithin{theorem}{section}

\newcommand\bB{{\mathbf B}}

\newcommand\cD{{\mathcal D}}

\newcommand\cG{{\mathcal G}}

\newcommand\cN{{\mathcal N}}

\newcommand\cT{{\mathcal T}}

\newcommand\scN{{\mathscr N}}

\newcommand\mvector{\boldsymbol}

\newcommand\vv{\mvector{v}}

\newcommand\vx{\mvector{x}}

\newcommand\vA{\mvector{A}}
\newcommand\vB{\mvector{B}}

\newcommand\vM{\mvector{M}}

\newcommand\vX{\mvector{X}}

\newcommand\vxi{\mvector{\xi}}

\newcommand\vGamma{\mvector{\Gamma}}

\newcommand\vvarphi{\mvector{\varphi}}

\newcommand\field{\mathbb}

\newcommand\bbS{\mathbb{S}}

\newcommand\C{\field{C}}
\newcommand\Z{\field{Z}}

\newcommand\N{\field{N}}

\newcommand\T{\field{T}}

\newcommand\ord{\operatorname{ord}}
\newcommand\grad{\operatorname{grad}}

\newcommand\rmd{\mathrm{d}}

\newcommand\rmi{\mathrm{i}\mspace{1mu}}

\newcommand\Dt{\frac{\mathrm{d}\phantom{t} }{\mathrm{d}\mspace{1mu}
t}}
\newcommand\DDt{\frac{\mathrm{d}^2\phantom{t} }{\mathrm{d}\mspace{1mu}
t^2}}

\newcommand\Dz{\frac{\mathrm{d}\phantom{z} }{ \mathrm{d}z}}
\newcommand\DDz{\frac{\mathrm{d}^2\phantom{z} }{ \mathrm{d}z^2}}

\newcommand\mtext[1]{\quad\text{#1}\quad}

\title{Non-integrability of flail  triple pendulum}
\author{
Maria Przybylska \\
Institute of Physics, University of Zielona G\'ora, \\
 Licealna 9, PL-65--417,  Zielona G\'ora, Poland
\\ e-mail: M.Przybylska@proton.if.uz.zgora.pl
\\ and \\
Wojciech Szumi\'nski\\
Institute of Physics, University of Zielona G\'ora, \\
 Licealna 9, PL-65--417,  Zielona G\'ora, Poland \\
e-mail: uz88szuminski@gmail.com 
}
\begin{document}
\mathtoolsset{%
mathic,centercolon%
}
\maketitle

\date{\small AMS Subject Classification: 70H07; 70H12;  70F07 }
\maketitle

\abstract{ We consider a special type of triple pendulum with two
  pendula attached to end mass of another one. Although we consider
  this system in the absence of the gravity, a quick analysis of of
  Poincar\'e cross sections shows that it is not integrable.  We give
  an analytic proof of this fact analysing properties the of
  differential Galois group of variational equation along certain
  particular solutions of the system.  }

\section{Introduction}
In this paper we study dynamics of a special type of triple pendulum
which we called a flail.  It consists of three pendula. The first one
is attached to a fixed point, and to its end mass the other two pendula
are joined.  The lengths and the masses of pendula are $l_1, l_2, l_3$, and
$m_1, m_2, m_3$, respectively, see Fig.~\ref{fig:lamane}.  We
consider this system in the absence of gravity. Similar problem of
dynamics of a simple triple pendulum in the absence of gravity field was
analysed in \cite{Salnikov:06::}.
\begin{figure}[h]
  \centering \includegraphics[width=0.38\textwidth]{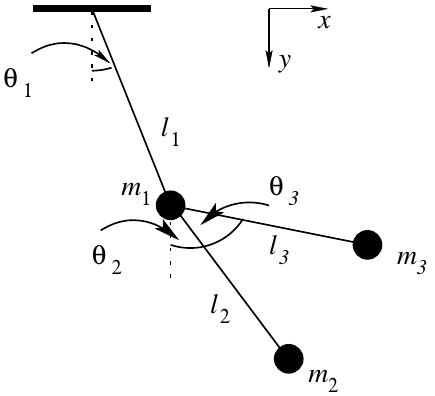}
  \caption{\footnotesize \label{fig:lamane}Geometry of the flail
    pendulum}
\end{figure}

The configuration space of our system is a torus $\T^3=\bbS^1\times
\bbS^1\times \bbS^1$ with local coordinates
$(\theta_1,\theta_2,\theta_3) \mod 2\pi$.  We chose the point of
suspension of the first pendulum as the origin, and angles are
measured from vertical line as it is shown in
Fig.~\ref{fig:lamane}. The Lagrange function has the following form
\begin{equation}
  \begin{split}
    L&=\dfrac{1}{2}\left(l_1^2 \mu\dot\theta_1^2 + l_2^2 m_2
      \dot\theta_2^2 + l_3^2 m_3 \dot\theta_3^2\right) +
    l_1 l_2 m_2 \cos(\theta_1-\theta_2)\dot\theta_1\dot\theta_2\\
    &+l_1 l_3 m_3
    \cos(\theta_1-\theta_3)\dot\theta_1\dot\theta_3,
      \end{split}
\end{equation} 
where $ \mu =m_1 + m_2 + m_3$.
Let us remark that the dynamics of this system can be interpreted as a
geodesic motion on torus $\T^3$.  The system has $\bbS^1$ symmetry.
This is why the Lagrange function depends on differences of
angles only. Hence, it is reasonable to introduce new variables defined by
\begin{equation}
  \gamma_1=\theta_2-\theta_1,\qquad \gamma_2=\theta_3-\theta_2,\qquad
  \gamma_3=\theta_1
  \label{eq:reduk}
\end{equation} 
and corresponding momenta
\[
\begin{split}
  &\pi_1=(l_2^2 m_2 + l_3^2 m_3)\dot\gamma_1+l_3^2
  m_3\dot\gamma_2+[l_2^2 m_2 + l_3^2 m_3 + l_1 l_2 m_2 \cos\gamma_1 +
  l_1 l_3 m_3 \cos(\gamma_1 +
  \gamma_2)]\dot\gamma_3,\\
  &\pi_2=l_3^2 m_3\dot\gamma_1+l_3^2 m_3\dot\gamma_2+l_3 m_3 [l_3 +
  l_1
  \cos(\gamma_1 + \gamma_2)]\dot\gamma_3,\\
  &\pi_3=[l_2^2 m_2 + l_3^2 m_3 + l_1 l_2 m_2 \cos\gamma_1 + l_1 l_3
  m_3 \cos(\gamma_1 + \gamma_2)]\dot\gamma_1+l_3 m_3 [l_3 + l_1
  \cos(\gamma_1 +
  \gamma_2)]\dot\gamma_2\\
  &+[l_2^2 m_2 + l_3^2 m_3 + l_1^2\mu + 2 l_1 l_2 m_2 \cos\gamma_1 +2
  l_1 l_3 m_3 \cos(\gamma_1 + \gamma_2)]\dot\gamma_3.
\end{split}
\]
Now, $\gamma_3$ is a cyclic variable, and $\pi_3$ is a first integral.
Thus, after this reduction, the system has two degrees of freedom and
its dynamics is determined by Hamiltonian
\begin{equation}
  \begin{split}
    H&= [l_1^2\mu (l_3^2 m_3 (\pi_1 - \pi_2)^2 + l_2^2 m_2 \pi_2^2) +
    l_2^2 l_3^2 m_2 m_3 (\pi_1 -c)^2\\
    & - l_1 (l_1 l_2^2 m_2^2 \pi_2^2 \cos\gamma_1^2 + l_3 m_3 \cos(
    \gamma_1 + \gamma_2) (2 l_2^2 m_2 \pi_2 (c-\pi_1)\\
    & + l_1 l_3 m_3 (\pi_1 - \pi_2)^2 \cos(\gamma_1 + \gamma_2)) - 2
    l_2 l_3 m_2 m_3 (\pi_1 - \pi_2) \cos\gamma_1 (l_3
    (\pi_1 -c) \\
    &+ l_1 \pi_2 \cos(\gamma_1 + \gamma_2)))]/[2 l_1^2 l_2^2 l_3^2 m_2
    m_3 B],
  \end{split}
  \label{eq:hamfin} 
\end{equation} 
where
\begin{equation*}
  \label{eq:2}
  B:=\mu -m_2 \cos^2\gamma_1 - m_3 \cos^2(\gamma_1 + \gamma_2), 
\end{equation*}
and $c:=\pi_3$ is a parameter.
\begin{figure}[p]
  \begin{center}
    \includegraphics[scale=0.165]{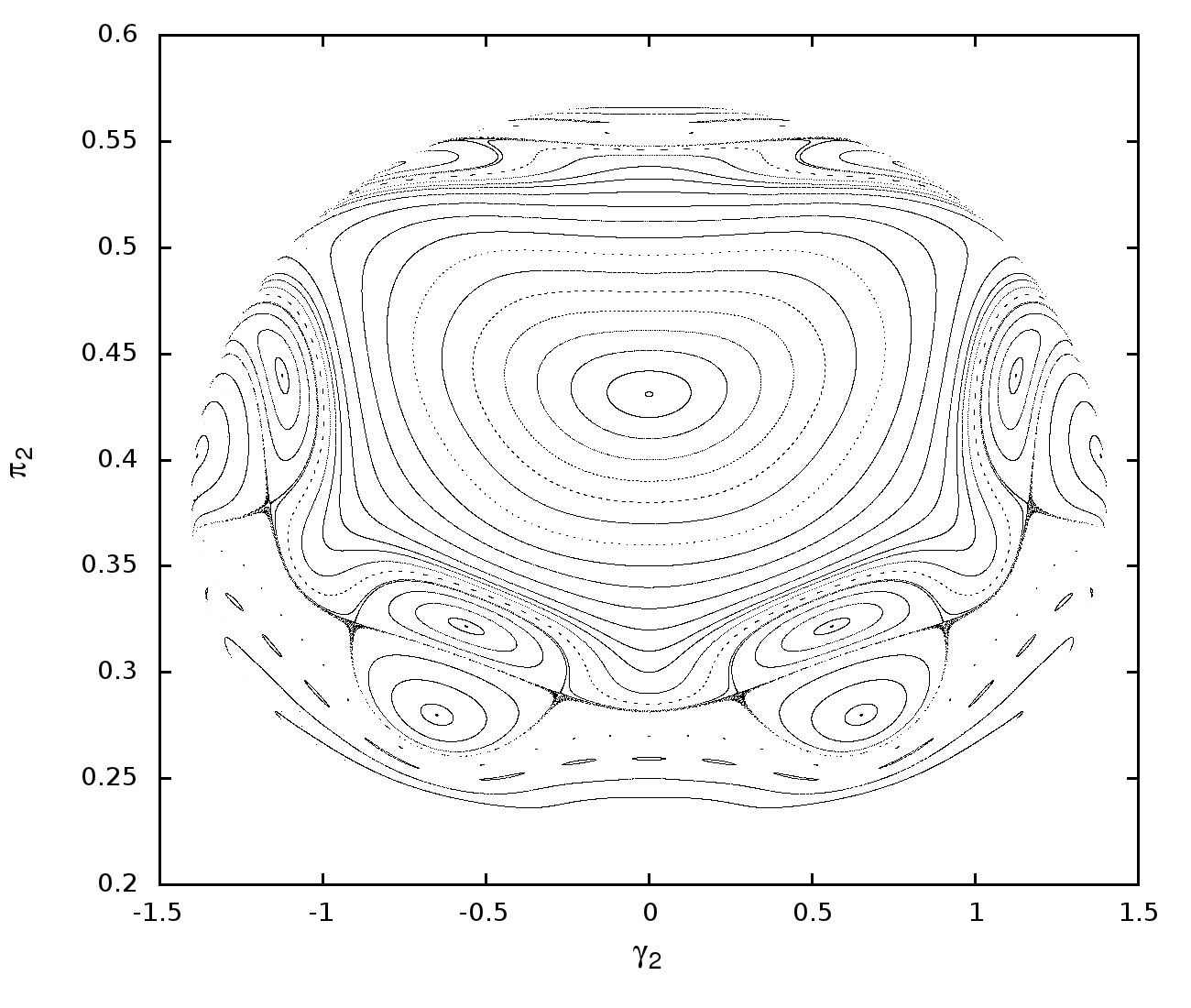} \includegraphics[scale=0.165]{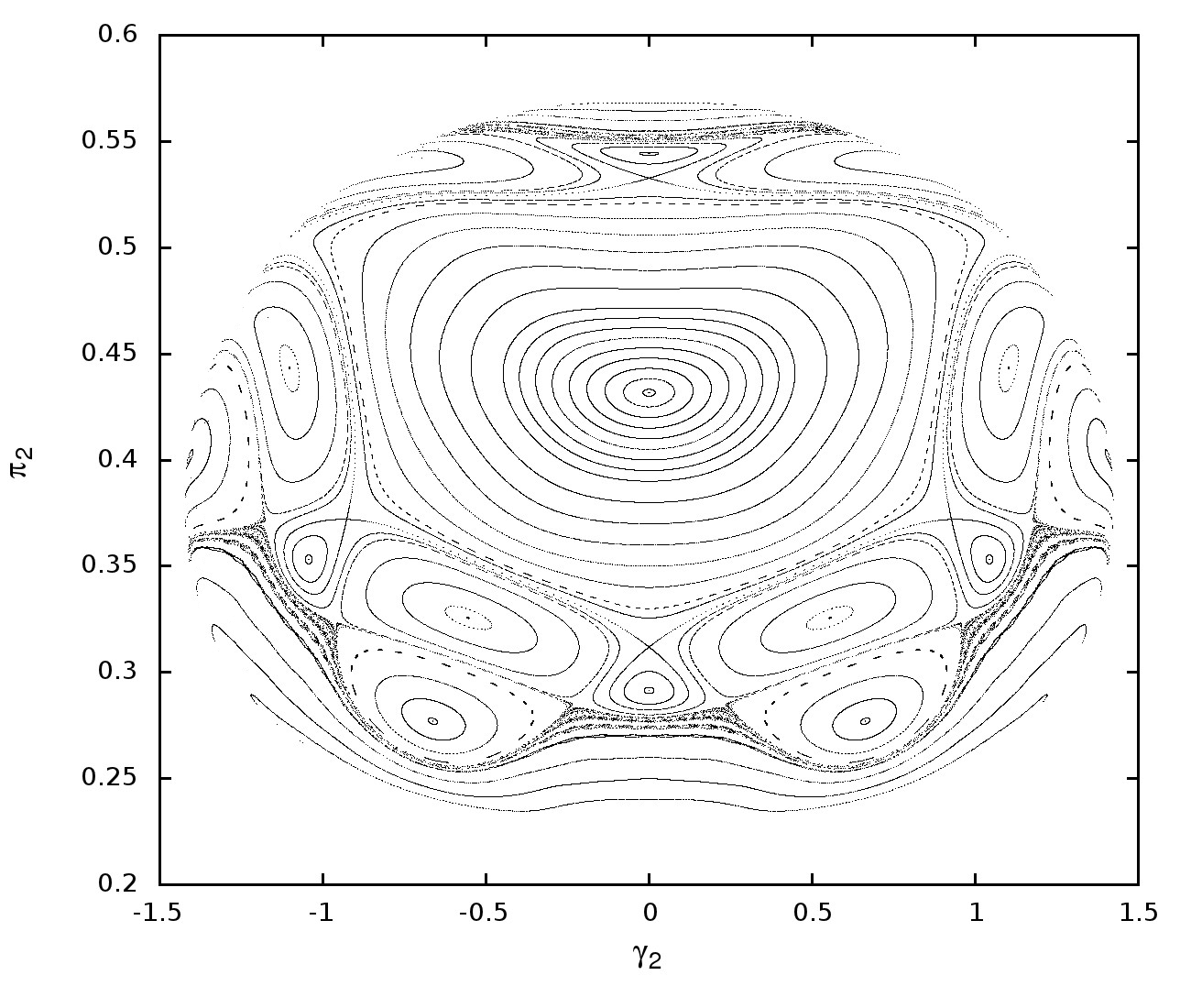}
\caption{\small The Poincar\'e  sections for $E = 0.01$ (on the left) and $E
= 0.01005$ (on the right)
\label{l1}}
\end{center}
\end{figure}

\begin{figure}[tp]
  \begin{center}
    \includegraphics[scale=0.165]{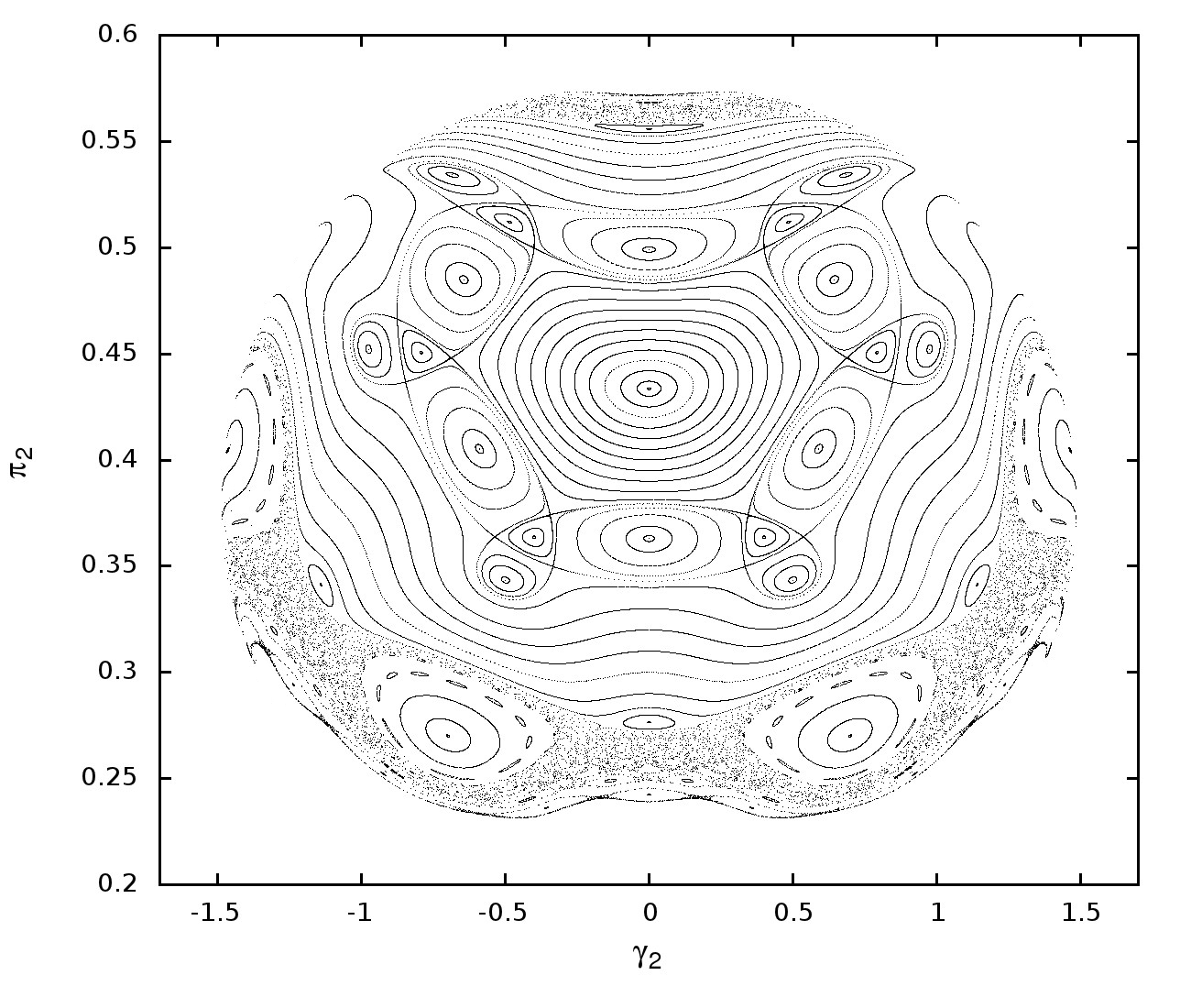} \includegraphics[scale=0.165]{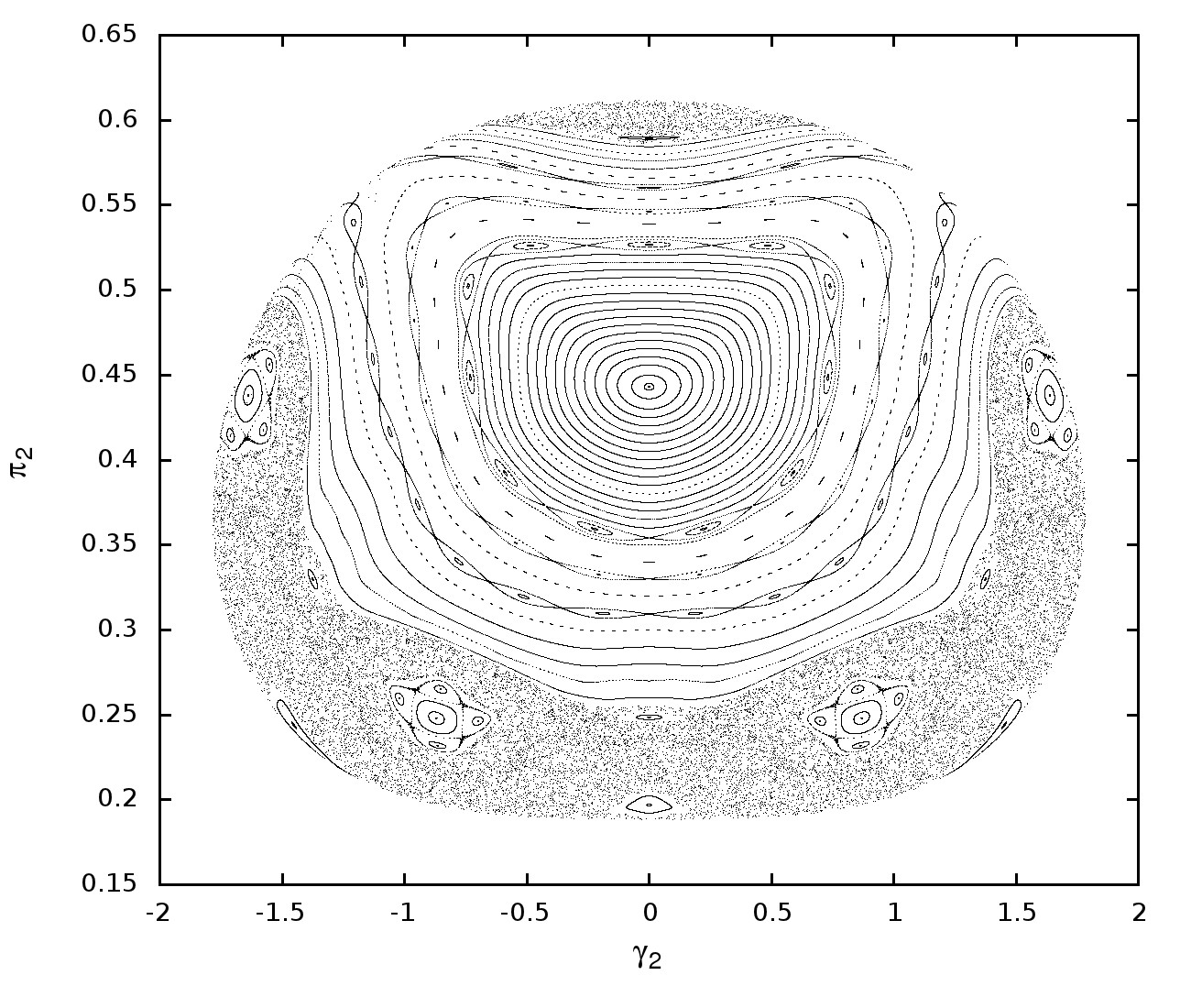}
    \caption{\small The Poincar\'e sections for $E = 0.0102$ (on the
      left) and $E = 0.011$ (on the right) \label{l3}}
  \end{center}
\end{figure}

\begin{figure}[htp]
  \begin{center}
    \includegraphics[scale=0.165]{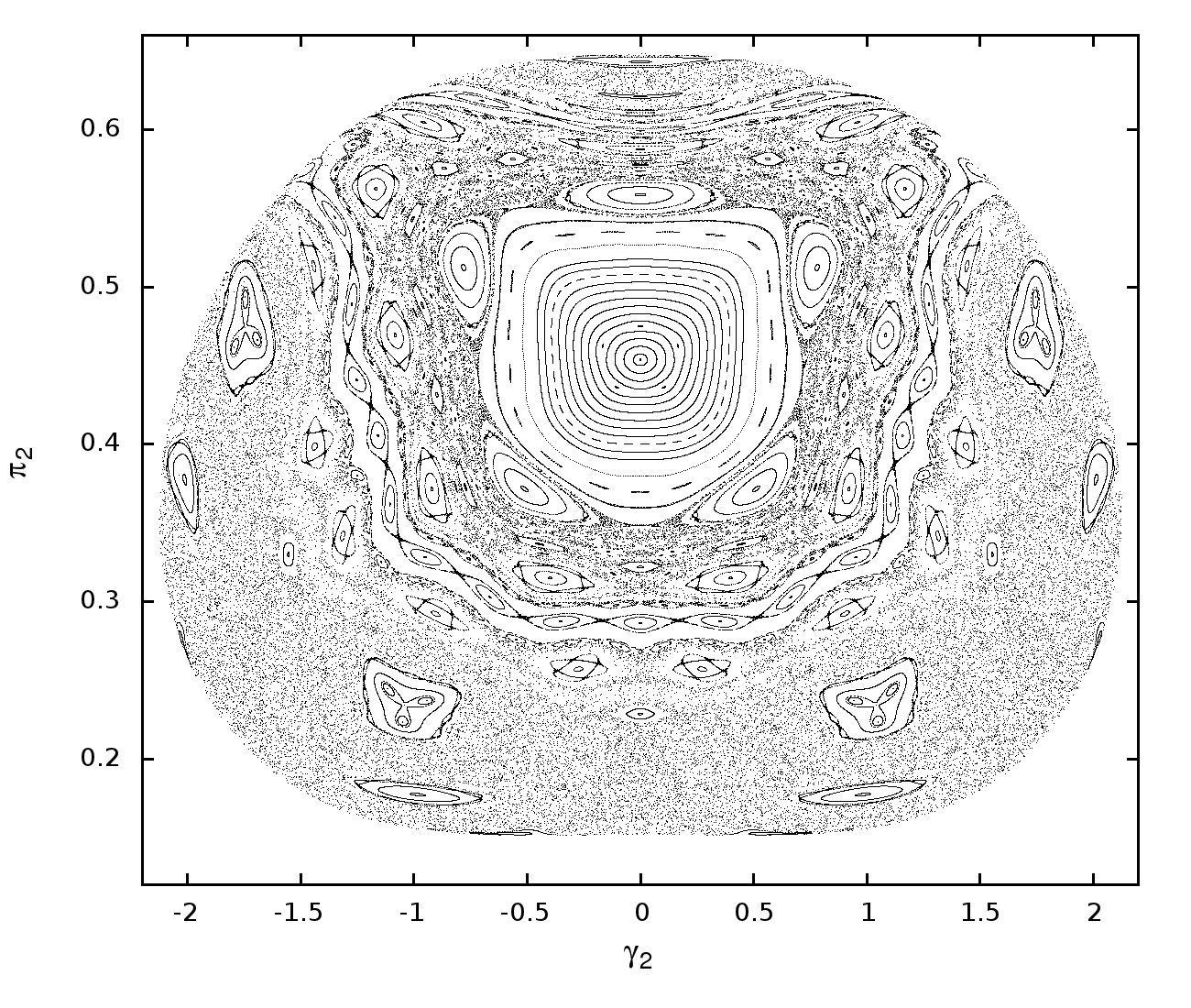}
    \includegraphics[scale=0.165]{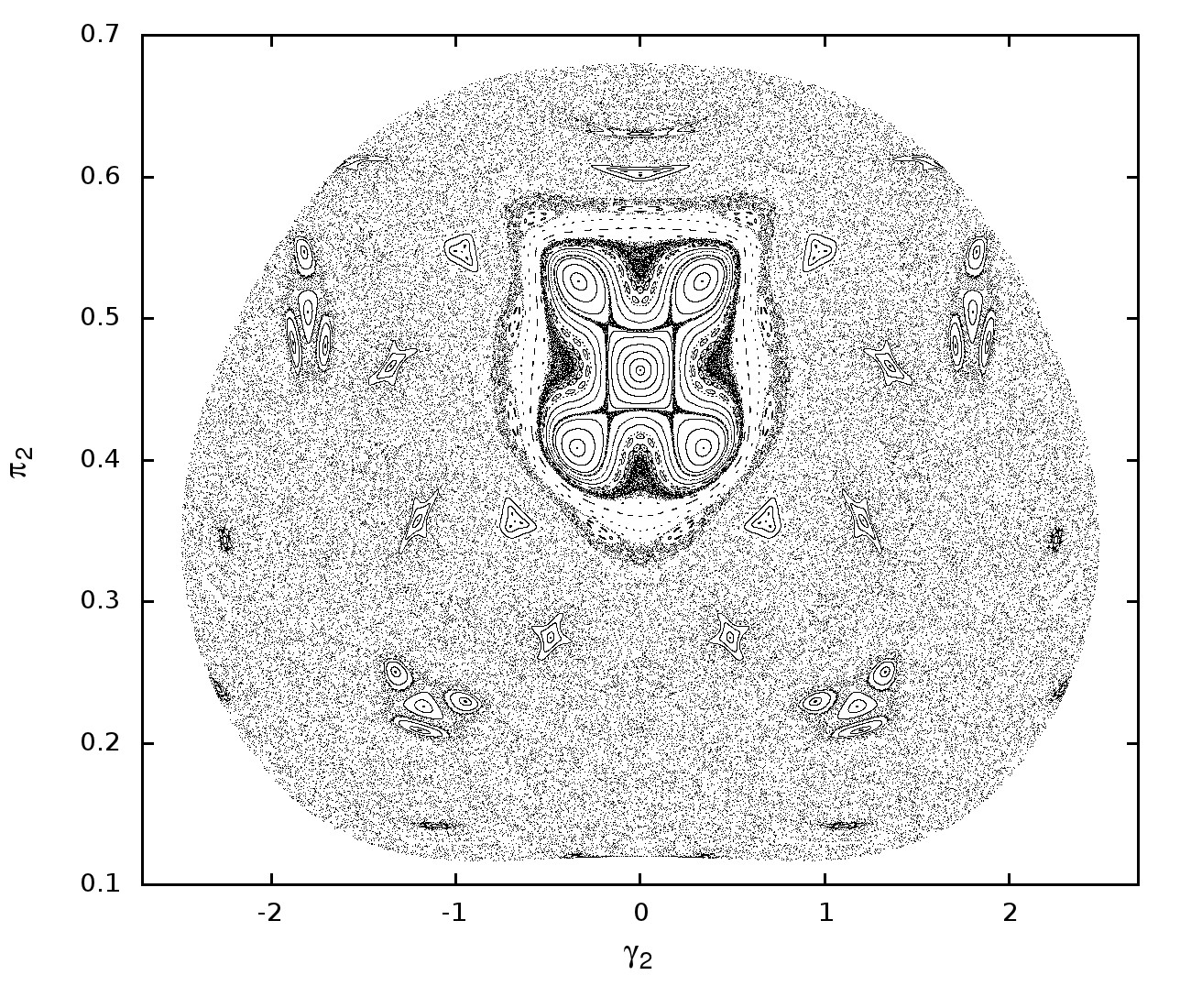}
    \caption{\small The Poincar\'e sections for $E = 0.012$ (on the
      left) and $E = 0.013$ (on the right) \label{l5}}
  \end{center}
\end{figure}

The aim of this paper is an analysis of integrability of this reduced
system.  To get an idea about the complexity of dynamics we made
several Poincar\'e cross sections. In Figures~\ref{l1}, \ref{l3}, and
\ref{l5} we show such cross sections. We take $\gamma_1=0$ and $
\pi_1>0$ as the cross section plane. The phase portraits are presented
in $(\gamma_2, y_2)$ plane. The cross sections are made for the
following values of parameters:
\begin{equation}
  \pi_3 = c = 1,\quad  m_1 = 1,\quad  m_2 = 3,\quad  m_3 = 2,\quad  l_1 = 1,\quad  l_2 = 2,\quad 
  l_3 = 3.
  \label{eq:paramy}
\end{equation}
Figures are ordered according to increasing energy of the system.
They show that, for chosen values of parameters, the system is not
integrable.  The main problem considered in this paper is to prove that
for a wide range of the parameters the system is in fact
non-integrable.  Our main result we formulate in the following
theorem.
\begin{theorem}
  If the reduced flail system governed by Hamiltonian
  \eqref{eq:hamfin} has parameters satisfying
  \begin{itemize}
  \item  $m_1\neq 0$, $c\neq 0$, $l_2\neq l_3$,  and  $m_2l_2=m_3l_3$; or
  \item $m_2=m_3$, $l_2=l_3$, and $c=0$,
  \end{itemize}
  then it is nonintegrable in the class of meromorphic functions of
  coordinates and momenta.
  \label{thm:main}
\end{theorem}

The plan of the paper is the following. Section~\ref{sec:proof} contains the
non-integrability proof and is divided into two subsections. In
Section~\ref{sec:concl} some remarks are given.  Explicit forms of some
complicated expressions are given in Appendix~\ref{sec:appendix}. For
convenience of readers in Appendix~\ref{app:kov} basic information about 
linear second order differential equation with rational
coefficients and Kovacic algorithm are given.
\section{Non-integrability proof}
\label{sec:proof}
In this section we prove Theorem~\ref{thm:main}   
using  the Morales-Ramis theory, see \cite{Baider:96::b,Morales:99::c}.  

We formulate the main theorem of this theory in the following simple
settings. Let us consider a complex Hamiltonian system. We assume that the underlying phase space is 
$\mathbb{C}^{2n}$, which is considered as a linear symplectic space 
 equipped with the canonical symplectic form
\[ 
\omega = \sum_{i=1}^n \rmd
q_i\wedge \rmd p_i,
\] 
where $ \vx=(q_1,\ldots, q_n,p_1,\ldots,p_n)$
are the canonical coordinates.  Let $H:\mathbb{C}^{2n} \rightarrow
\mathbb{C}$ be a holomorphic Hamiltonian,  and
\begin{equation}
\label{eq:ds}
   \Dt \vx = \vv_H(\vx),\quad
\vv_H(x)=\mathbb{I}_{2n}\nabla_{\vx} H,
 \qquad \vx\in\mathbb{C}^{2n}, 
\quad t\in\mathbb{C},
\end{equation}
the associated Hamilton equations. Here $\mathbb{I}_{2n}$ denotes  $2n\times 2n$
symplectic  unit matrix.

Let  $t\mapsto\vvarphi(t)\in\C^{2n}$ be 
a non-equilibrium solution of \eqref{eq:ds}. The maximal analytic
continuation of $\vvarphi(t)$ defines a Riemann surface $\vGamma$ with
$t$ as a local coordinate. Variational equations  along
$\vvarphi(t)$ have the form 
\begin{equation}
\label{eq:ve}
 \Dt\vxi = \vA(t) \vxi, \qquad  
\vA(t) = \dfrac{\partial \vv_H}{\partial \vx}
(\vvarphi(t)). 
\end{equation}
We can attach to this equations  the differential Galois group $\cG$.
Roughly speaking, $\cG$ is an algebraic subgroup of
$\mathrm{GL}(2n,\C)$, which preserves  polynomial relations
between solutions of \eqref{eq:ve}.  As a linear algebraic group,
among other things, it is a union of a finite number of disjoint
connected components. One of them, containing the identity, is called
the identity component of $\cG$, and is denoted by $\cG^0$.  For a
precise definition of the differential Galois group and differential
Galois theory see e.g.,
\cite{Kaplansky:76::,Beukers:92::,Magid:94::,Put:03::}.  
It appears that the integrability of the considered system manifests
itself in properties of the differential Galois group of the
variational equations.  For Hamiltonian system and integrability in
the Liouville sense this relation is particularly elegant.  For these
systems $\cG$ is an algebraic subgroup of $\mathrm{Sp}(2n,\C)$. In
nineties of XX century Morales-Ruiz and Ramis showed that the
integrability in the Liouville sense imposes a very restrictive
condition on $\cG^0$ \cite{Baider:96::b,Morales:01::b1,Morales:99::c}.
\begin{theorem}[Morales-Ruiz and Ramis]
\label{thm:MR}
  Assume that a Hamiltonian system is meromorphically integrable in
  the Liouville sense in a neigbourhood of a phase curve
  $\vGamma$. Then the identity component of the differential Galois
  group of NVEs associated with $\vGamma$   is Abelian.
\end{theorem}    
In order to apply this theorem we need an effective method which
allows to determine properties of the differential Galois group of
linear equations.  In the investigated system variational equations
split into two subsystems of linear equations.  Each of these
subsystems can be transformed into equivalent second order equation
with rational coefficients. For such equation there exists the Kovacic
algorithm \cite{Kovacic:86::}, see also Appendix~\ref{app:kov}, which allows to
determine the
differential Galois group. We apply this algorithm in our considerations. 

\subsection{Case $m_2\neq m_3$ and $l_2\neq l_3$}
\label{sec:nierowne}
At first we make the following non-canonical transformation
\begin{equation}
  \begin{bmatrix}
    \gamma_1\\
    \gamma_2\\
    \pi_1\\
    \pi_2
  \end{bmatrix}
  =\vA
  \begin{bmatrix}
    x_1\\
    x_2\\
    y_1\\
    y_2
  \end{bmatrix},\qquad \vA=\begin{bmatrix}
    0&1&0&0\\
    1&-2&0&0\\
    0&0&1&1-\dfrac{l_2}{l_3}\\
    0&0&0&1
  \end{bmatrix}.
  \label{eq:trafononcan}
\end{equation}
After this transformation equations of motion take the form
\begin{equation}
  \dot\vx=\mathbb{J}_{2n}\nabla_{\vx}\widetilde{H},\qquad
  \vx=[x_1,x_2,y_1,y_2]^T,
  \label{eq:systfin}
\end{equation}
where
\begin{equation}
  \mathbb{J}_{2n}:=\bB\mathbb{I}\vB^T,\qquad \vB=\vA^{-1}, \quad
  \widetilde{H}(\vx):= H(\vA\vx).
\end{equation} 
The explicit form of Hamiltonian in these variable is following
\begin{equation}
  \label{eq:tH}
  \begin{split}
    &\widetilde H= \left[l_2^2 m_2 m_3\left (c l_3 + l_2 y_2 - l_3
        (y_1 + y_2)\right)^2 +
      l_1^2\mu \left(l_2^2 m_2 y_2^2 + m_3 (l_3 y_1 - l_2 y_2)^2\right) \right.\\
    & - l_1\left( m_3 \cos(x_1 - x_2) \left(2 l_2^2 m_2 y_2 \left(c
          l_3 + l_2 y_2 - l_3 (y_1 + y_2)\right)
        + l_1 m_3 (l_3 y_1 - l_2 y_2)^2 \cos(x_1 - x_2) \right) \right.\\
    & + 2 l_2 m_2 m_3 (-l_3 y_1 + l_2 y_2) \left(-l_2 y_2 + l_3 (-c +
      y_1 + y_2) +
      l_1 y_2 \cos(x_1 - x_2)\right) \cos x_2\\
    & + \left. \left.  l_1 l_2^2 m_2^2 y_2^2 \cos
        x_2^2\right)\right]/[2 l_1^2 l_2^2 l_3^2 m_2 m_3 (\mu - m_3
    \cos(x_1 - x_2)^2 - m_2 \cos x_2^2)].
  \end{split}
\end{equation} 
An explicit form of vector field
$\vv(\vx):=\mathbb{J}_{2n}\nabla_{\vx}\widetilde{H}$ is given in
Appendix~\ref{sec:appendix}, see equation~\eqref{eq:ukk}. These
equations are still Hamiltonian but written in non-canonical
variables.

Let us assume that the following conditions
\begin{equation}
  \label{eq:cond}
  m_2l_2=m_3l_3, \qquad l_2\neq l_3, 
\end{equation}
are fulfilled.  Then, system~\eqref{eq:systfin} has an invariant
manifold
\[
\cN=\{(x_1,x_2,y_1,y_2)\in\C^4\ |\ x_1=y_1=0,\, y_2=-cl_3/(l_2-l_3)
\}.
\]
In fact, restricting the right hand sides of
system~\eqref{eq:systfin} to $\cN$, and putting $ m_3={m_2l_2}/{l_3}
$, we obtain
\begin{equation}
  \dot x_1=0,\qquad \dot x_2=\dfrac{c}{l_2m_2(l_2-l_3)},\qquad \dot y_1=0,\qquad
  \label{eq:particular}
  \dot y_2=0. 
\end{equation} 
Hence, solving the above equations we obtain our particular solution
$\C\ni t\mapsto\vvarphi(t)$.

Let $\vX:=[X_1,X_2,Y_1,Y_2]^T$ denote variations of
$[x_1,x_2,y_1,y_2]^T$.  The variational equations along $\vvarphi(t)$
take the form
\begin{equation}
  \Dt\vX =\vA \vX, \qquad \vA=
  \begin{bmatrix}
    a_{11}&0&a_{13}&a_{14}\\
    a_{21}&0&a_{23}&a_{24}\\
    a_{31}&0&a_{33}&a_{34}\\
    a_{41}&0&a_{43}&a_{44}
  \end{bmatrix}.
  \label{eq:wariaty1}
\end{equation}
The explicit forms of entries $a_{ij}$ are given in
equation~\eqref{eq:wariaty1_entries} in Appendix~\ref{sec:appendix}.

We note that equations on $X_1$, $Y_1$ and $Y_2$ do not depend on
$X_2$, i.e., variable $X_2$ decouples from remaining
variables. Moreover one can observe that equations for $Y_1$ and for
$Y_2$ are dependent, see appropriate coefficients $a_{ij}$.  This fact
is connected with the existence of a first integral of variational
equations. The expansion of energy integral along the particular
solution gives the first integral of variational equations
\begin{equation}
  h:=\grad \widetilde H(\vvarphi)\cdot\vX=\dfrac{c}{l_2 l_3
    m_2(l_3-l_2)}[-l_3Y_1+(l_2+l_3)Y_2].
  \label{eq:variat}
\end{equation} 
At the level $h=0$ we have
\begin{equation}
  Y_2=\dfrac{l_3}{l_2+l_3}Y_1.
\end{equation}
Then, substituting this into the sub-system of variational equations
for $X_1$, $Y_1$ and $Y_2$ and eliminating $Y_1$ we obtain the
following second order linear equation for $X_1$
\begin{equation}
  \ddot X_1+a\dot X_1+bX_1=0,
  \label{eq:variaty}
\end{equation} 
with coefficients
\[
\begin{split}
  &a=[4 c (l_2 + l_3) (2 l_2 l_3 + l_1 (l_2 + l_3) \cos x_2) (l_1 l_3
  m_1 +
  l_1 (l_2 + l_3) m_2 + 2 l_2 l_3 m_2 \cos x_2) \sin x_2]\\
  &/[l_2 (l_2 - l_3) (4 l_2^2 l_3^2 m_2 +l_1^2 (l_2 + l_3) (l_3 m_1 +
  (l_2 + l_3) m_2) +
  4 l_1 l_2 l_3 (l_2 + l_3) m_2 \cos x_2) (2 l_3 m_1\\
  & + (l_2 + l_3) m_2 -
  (l_2 + l_3) m_2 \cos(2x_2))],\\
  &b=[2 c^2 (2 l_2 l_3 +l_1 (l_2 + l_3) \cos x_2) (4 l_1 l_2 l_3 (l_2
  + l_3) m_2 +
  (4 l_2^2 l_3^2 m_2 +l_1^2 (l_2 + l_3) (l_3 m_1\\
  & + (l_2 + l_3) m_2)) \cos x_2)]/[l_1 l_2^2 (l_2 - l_3)^2 m_2 (4
  l_2^2 l_3^2 m_2 +l_1^2 (l_2 + l_3)
  (l_3 m_1 + (l_2 + l_3) m_2)\\
  & + 4 l_1 l_2 l_3 (l_2 + l_3) m_2 \cos x_2) (2 l_3 m_1 + (l_2 + l_3)
  m_2 - (l_2 + l_3) m_2 \cos(2x_2))].
\end{split}
\]
In order to transform this equation into an equation with rational
coefficients we use the following transformation
\begin{equation}
  t\longmapsto z=\cos x_2(t).
  \label{eq:ratio}
\end{equation}
This change of variable together with transformations of derivatives
\[
\begin{split}
  &\Dt=\dot z\Dz,\qquad \DDt=(\dot z)^2\DDz+\ddot z \Dz,\\
  &\dot z=-\sin x_2\, \dot x_2=-\dfrac{c \sin
    x_2}{l_2m_2(l_2-l_3)},\qquad
  (\dot z)^2=\dfrac{c^2 (1-z^2)}{l_2^2m_2^2(l_2-l_3)^2},\\
  &\ddot z=-\cos x_2\, \dot x_2^2=-\dfrac{c^2
    z}{l_2^2m_2^2(l_2-l_3)^2},
\end{split}
\]
convert equation~\eqref{eq:variaty} into
\begin{equation}
  X''+p(z)X'+q(z)X=0,  \qquad X=X_1,
  \label{eq:zvar}
\end{equation} 
where the prime denotes the derivative with respect to $z$. The
explicit forms of coefficients $p(z)$ and $q(z)$ are the following
\[
\begin{split}
  & p=\dfrac{1}{\dot z^2}(\ddot z+\dot z a)=\dfrac{z}{z^2-1} +
  \dfrac{2\sigma z}{(z^2-1) \sigma - \alpha \gamma} - \dfrac{4
    \sigma}{ 4 + \sigma (4 z + \sigma + \alpha \gamma)},\\
  &q=\dfrac{b}{\dot z^2}=\dfrac{(\sigma z +2) (4 \sigma+z (4 + \sigma
    (\sigma + \alpha \gamma)))}{(z^2-1)[\sigma (z^2-1) - \alpha
    \gamma) (4 +\sigma (4 z + \sigma + \alpha \gamma)]},
\end{split}
\]
where we introduced notation
\[
\alpha=\dfrac{l_1}{l_2},\qquad \beta=\dfrac{l_1}{l_3},\qquad
\gamma=\dfrac{m_1}{m_2}, \qquad \sigma:=\alpha+\beta.
\]
These parameters take non-negative values, and we assume that they do
not vanish simultaneously.

Now, we make the classical Tschirnhaus change of dependent variable,
\begin{equation}
  \label{eq:tran}
  X = w \exp\left[ -\frac{1}{2} \int_{z_0}^z p(\zeta)\,  \rmd \zeta \right]
\end{equation}
which transforms \eqref{eq:zvar} into its reduced form
\begin{equation}
  \label{eq:normal}
  w'' = r(z) w, \qquad r(z) = -q(z) + \frac{1}{2}p'(z)  + \frac{1}{4}p(z)^2,
\end{equation}
with coefficient
\[
\begin{split}
  r(z)= & -\frac{3}{16}\left[\frac{1}{\left( z-1 \right)^{2}}
    +\frac{1}{\left( z+1 \right)^{2}}\right]
  -\frac{1}{4}\left[\frac{1}{\left( z-\delta \right)^{2}}
    +\frac{1}{\left( z+\delta\right)^{2}} -\frac{3}{\left( z-\rho
      \right)^2}\right] + \\& \frac{\sigma^3\delta^2 \left( 5z^2
      -\delta^2 -4\right) +4\sigma^2z \left( 4+\delta^2 -z^2 \right) +
    4\sigma\left( 20+\delta^2 -5z^2 \right) +64z}{32\sigma^2
    (z^2-1)(z^2-\delta^2)(z-\rho)},
\end{split}
\]
where
\begin{equation}
  \label{eq:sdr}
  \delta^2 :=1+\frac{\alpha\gamma}{\sigma}, 
  \qquad \rho:=- \frac{1}{\sigma}-\frac{1}{4}\delta^2\sigma.
\end{equation}
Let us notice that coefficient $r(z)$ depends only on two parameters
$\delta$ and $\sigma$. We assume that  both are positive and
$\delta\geq 1$.

If $\delta\neq 1$, and $\sigma\delta\neq 2$, then
equation~\eqref{eq:normal} has six pairwise different regular singular
points
\[
z_{1,2}=\pm 1,\qquad z_{3,4}=\pm\delta,\qquad z_5=\rho,\qquad
z_6=\infty.
\]
The respective differences of exponents at these points are following
\begin{equation}
  \label{eq:de}
  \Delta_1=\Delta_2=\frac{1}{2}, \qquad 
  \Delta_3=\Delta_4=0, \qquad 
  \Delta_5=2, \qquad \Delta_6=1.
\end{equation}
Now, we can prove the following result.
\begin{lemma}
 For
  \begin{equation}
    \delta\neq 1,  \mtext{and} \delta\sigma\neq2,
    \label{eq:except}
  \end{equation}
  the differential Galois group of equation \eqref{eq:normal} 
  is $\mathrm{SL}(2, \C)$.
  \label{lem:gennier}
\end{lemma}
\begin{proof}
  We prove this lemma by a contradiction. Thus, let us assume that $
  \cG\neq \mathrm{SL}(2, \C)$. Then, according to Lemma~\ref{lem:alg},
  $\cG$ is either a finite subgroup of $\mathrm{SL}(2, \C)$, or a
  subgroup of dihedral group or a triangular group of $\mathrm{SL}(2,
  \C)$.

  At first we show that the first two possibilities do not occur.  As
  $\Delta_{3,4}=0$, local solutions in a neighbourhood of
  $z_{\star}=z_{3}$, and $z_{\star}=z_{4}$, contain a logarithmic
  term. More precisely two linearly independent solutions $w_1$ and
  $w_2$ of \eqref{eq:normal} in a neighbourhood of
  $z_{\star}\in\{z_3,z_4\}$, have the following forms
  \begin{equation}
    \label{eq:w12}
    w_1(z) = (z - z_{\ast})^{\rho} f(z), \qquad w_2(z)=w_1(z)\ln
    (z-z_{\ast})+(z-z_{\ast})^{\rho} h(z),
  \end{equation}
  where $f(z)$ and $ h(z)$ are holomorphic at $z_{\ast}$, and
  $f(z_{\ast})\neq 0$. The local monodromy matrix corresponding to
  a continuation of solutions along a small loop encircling $z_{\ast}$
  counterclock-wise gives rise to a triangular monodromy matrix
  \[
  \vM_{\star}=
  \begin{bmatrix*}[r]
    -1&-2\pi\rmi\\
    0&-1
  \end{bmatrix*},
  \]
  for details, see \cite{mp:02::d}.  A subgroup of $\mathrm{SL}(2,
  \C)$ generated by $\vM_{\star}$ is not finite, and thus also the
  differential Galois group $\cG$ of~\eqref{eq:normal} is not
  finite. As $\vM_{\star}$ is not diagonalisable, $\cG$ is not a
  subgroup of dihedral group.  Thus, $\cG$ is a subgroup of the
  triangular group of $\mathrm{SL}(2, \C)$.

  In order to check if it is really the case we apply the Kovacic algorithm,
for details
  see Appendix~\ref{app:kov}. According to it, if $\cG$ is a subgroup of the
  triangular group of $\mathrm{SL}(2, \C)$, the first case of this
  algorithm occurs, i.e., the equation admits an exponential solution.

  To check this possibility we determine sets of exponents
  \[
  E_i:= \{ (1\pm\Delta_i)/2\}, \qquad\text{for}\quad i=1,\ldots,6.
  \]
  Thus we have
  \begin{equation}
    E_1=E_2=\left\{\dfrac14,\dfrac34\right\},\quad
    E_3=E_3=\left\{\dfrac12\right\},\quad
    E_5=\left\{-\dfrac12,\dfrac32\right\},\quad E_6=\{0,1\}.
  \end{equation}
  Next, according to the algorithm  we look for elements
  $e=(e_1,e_2,e_3,e_4,e_5,e_6)$ of Cartesian product
  $E=\prod_{i=1}^6E_i$ such that
  \begin{equation}
    d(e):=e_6-\sum_{i=1}^5e_i\in\N_0,
  \end{equation}
  where $\N_0$ is the set of non-negative integers.  In our case,
  there exists only one $e\in E$ with this property, namely
  \begin{equation}
    e=\left(\dfrac14,\dfrac14,
      \dfrac12,\dfrac12,-\dfrac12,1\right), 
  \end{equation}
  for which $d(e)=0$. Thus, according to the algorithm, we look for a
  polynomial $P\neq0$ of degree $d(e)$, such that it is a solution of
  the following equation
  \begin{equation}
    \label{eq:P}
    P'' +2 \omega P' + (\omega' + \omega^2 - r) P = 0,
  \end{equation}
  where
  \[
  \omega=\sum_{i=1}^5\dfrac{e_i}{z-z_i}.
  \]
  In the considered case, we have $P=1$, so equation \eqref{eq:P}
  gives equality
  \[
  (\sigma z +2) \left[ (\sigma^2 \delta^2+4)z+4\sigma\right]=0,
  \]
  which cannot be fulfilled for $\sigma>0$. This end the proof.
\end{proof}
Now, let us analyse the cases excluded in Lemma~\ref{lem:gennier}.
Case $\sigma=0$ is equivalent to $\alpha=l_1/l_3=0$ and
$\beta=l_1/l_2=0$. So, it is equivalent to the case of zero length of
the first pendulum. In effect, the system consists of two simple
independent pendula and is obviously integrable.

In the case of confluence of singular point $z_3$ and $z_5$ we prove
the following.
\begin{lemma}
  Let us assume that
  \begin{equation}
    \label{eq:c35}
    \delta=\frac{2}{\sigma}\geq 1. 
  \end{equation}
  Then the differential Galois group of equation \eqref{eq:normal} is
  $\mathrm{SL}(2,\C)$.
  \label{lem:niegennier}
\end{lemma}
\begin{proof}
  Then coefficient $r$ in equation~\eqref{eq:normal} simplifies to the
  following form
  \begin{equation}
    \begin{split}
      r=& -\dfrac{3}{16 (z-1)^2} -\dfrac{3}{16 (z+1)^2}
      -\dfrac{1}{4(z-z_3)^2}-
      \dfrac{1}{2\rho(z-z_4)} +\\
      &\dfrac{3-\rho}{ 16 (\rho+1) (z+1)} +\dfrac{3+\rho}{ 16 (\rho-1)
        (z-1)} -\dfrac{1}{2\rho (\rho^2-1) (z-z_{3})},
    \end{split}
    \label{eq:rrexp}
  \end{equation} 
  where $z_3=-z_4=\delta:=2/\sigma$. Now, if $\delta>1$, then equation
  \eqref{eq:normal} has five singularities $z_{1,2}=\pm 1$, $z_3$ and
  $z_4$, and $z_5=\infty$.

  One can easily check that the difference of exponents at
  $z_{\star}=z_3$ vanishes.  Thus the local solutions around this
  point have the form~\eqref{eq:w12}.  One can repeat reasoning from
  the proof of Lemma~\ref{lem:gennier}.  In order to check whether the first
case of the Kovacic algorithm can appear  at first we calculate the  sets of
exponents $E_i$ for $i=1,\ldots,5$
  \begin{equation}
    E_1=E_2=\left\{\dfrac14,\dfrac34\right\},\quad
    E_3=\left\{\dfrac12\right\},\quad
    E_4=\left\{0,1\right\},\quad E_5=\{0,1\}.
  \end{equation}
  One can easily check that for each $e=(e_1,e_2,e_3,e_4, e_5 )\in
  \prod_{i=1}^5E_i$ we have
  \begin{equation}
    d(e)=e_5-\sum_{i=1}^4e_i\not\in\N_0.
  \end{equation}
  Thus, if $\delta>1$, according to Kovacic algorithm, the differential
  Galois group of \eqref{eq:normal} is $\mathrm{SL}(2,\C)$.

  For $\delta=1$, i.e., for $\sigma=2$, coefficient $r$ simplifies to
  \begin{equation}
    r(z)=-\dfrac{7 + 4 z}{4 (-1 + z^2)^2}.
    \label{eq:redek}
  \end{equation}
  Now, equation \eqref{eq:normal} has only three regular
  singularities, i.e., it is the Riemann $P$ equation. Using the well
  known Kimura theorem~\cite{Kimura:69::}, one can check that its
  differential Galois group is $\mathrm{SL}(2,\C)$.
\end{proof}
\begin{remark}
  An equivalent form of condition $\sigma\delta=2$ is following
  \[
  l_1^2 (l_2 + l_3) (l_2 m_2 + l_3 (m_1 + m_2))-4 l_2^2 l_3^2 m_2=0.
  \]
\end{remark}
For $m_1=0$  and $l_1\neq 0$, using the Kovacic algorithm,  one can  prove only
that if the differential Galois group of variational equations along the
mentioned particular solution is not $\mathrm{SL}(2,\C)$, then it  
 is a finite subgroup of $\mathrm{SL}(2,\C)$. Calculations in first and
second case of this algorithm are simple but in the third case becomes really
complicated.   Poincar\'e
sections given in  Fig.~\ref{l7} suggest that also in this case system is
not integrable. 
\begin{figure}[htp]
  \begin{center}
    \includegraphics[scale=0.16]{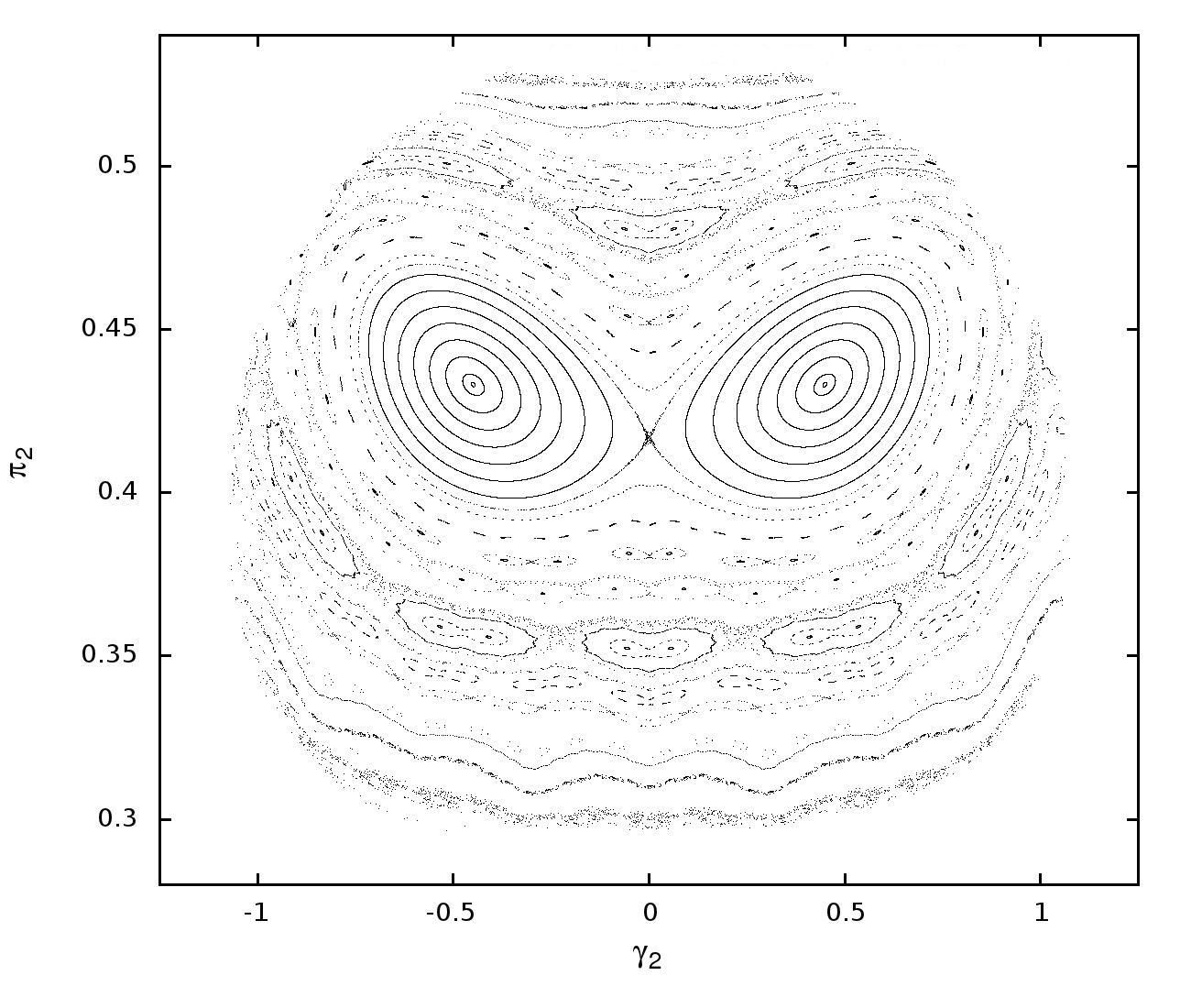}
    \includegraphics[scale=0.16]{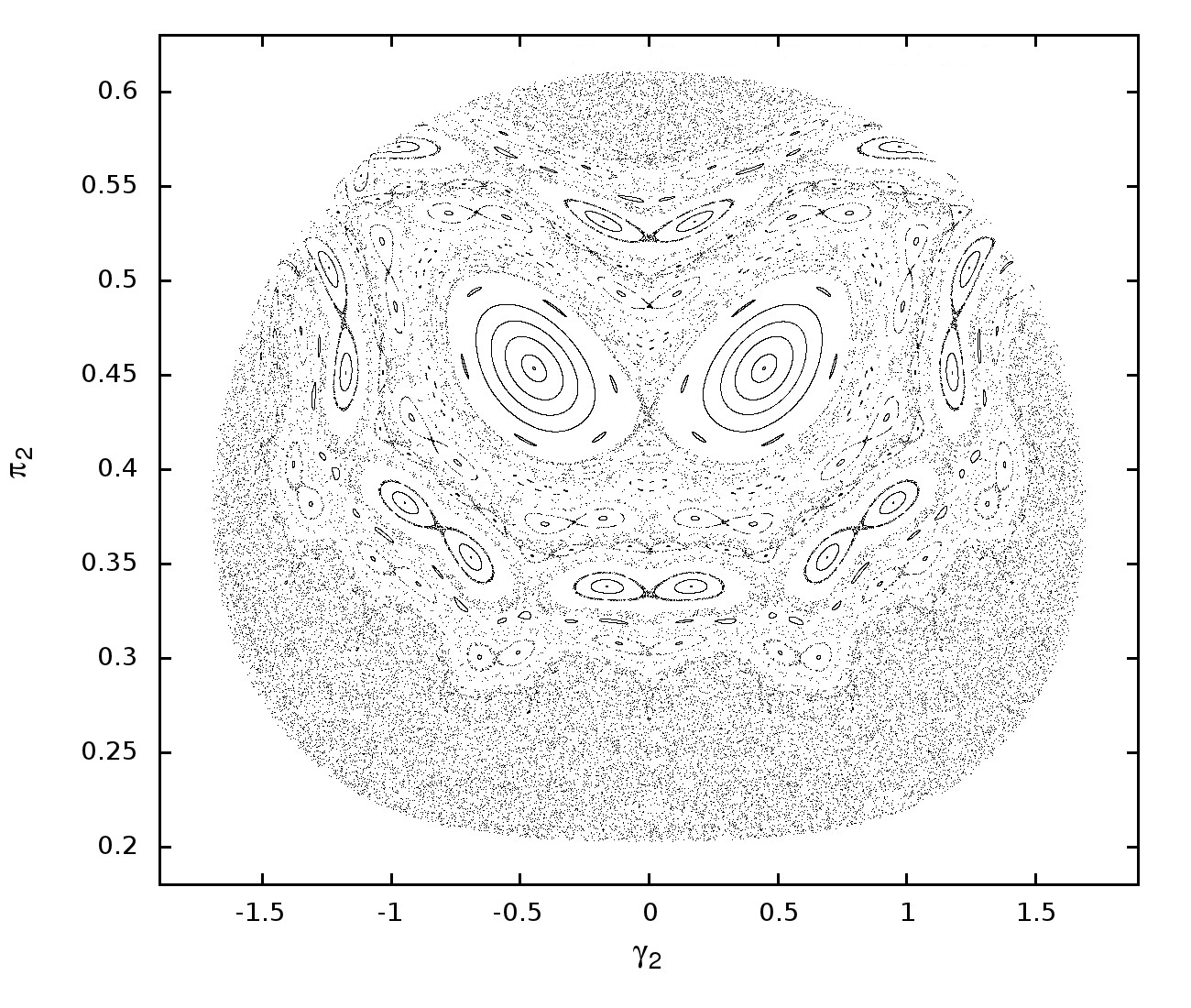}
  \end{center}
  \caption{\small Poincar\'e sections for $m_1 = 0,\, m_2 = 3,\, m_3 =
    2,\, l_1 = 1,\, l_2 = 2,\, l_3 = 3,\, \pi_3 = c = 1,$ and for $E =
    0.0095$ (on the left) and $E = 0.011$ (on the
    right). Cross-section plane $\gamma_1=0$ \label{l7}}
\end{figure}

The proof of the first part of Theorem~\ref{thm:main} is an immediate
consequence of Lemma~\ref{lem:gennier}, \ref{lem:niegennier} and
Theorem~\ref{thm:MR}.
\subsection{Case $l_3=l_2$ and $m_3=m_2$}
\begin{figure}[htp]
  \begin{center}
    \includegraphics[scale=0.16]{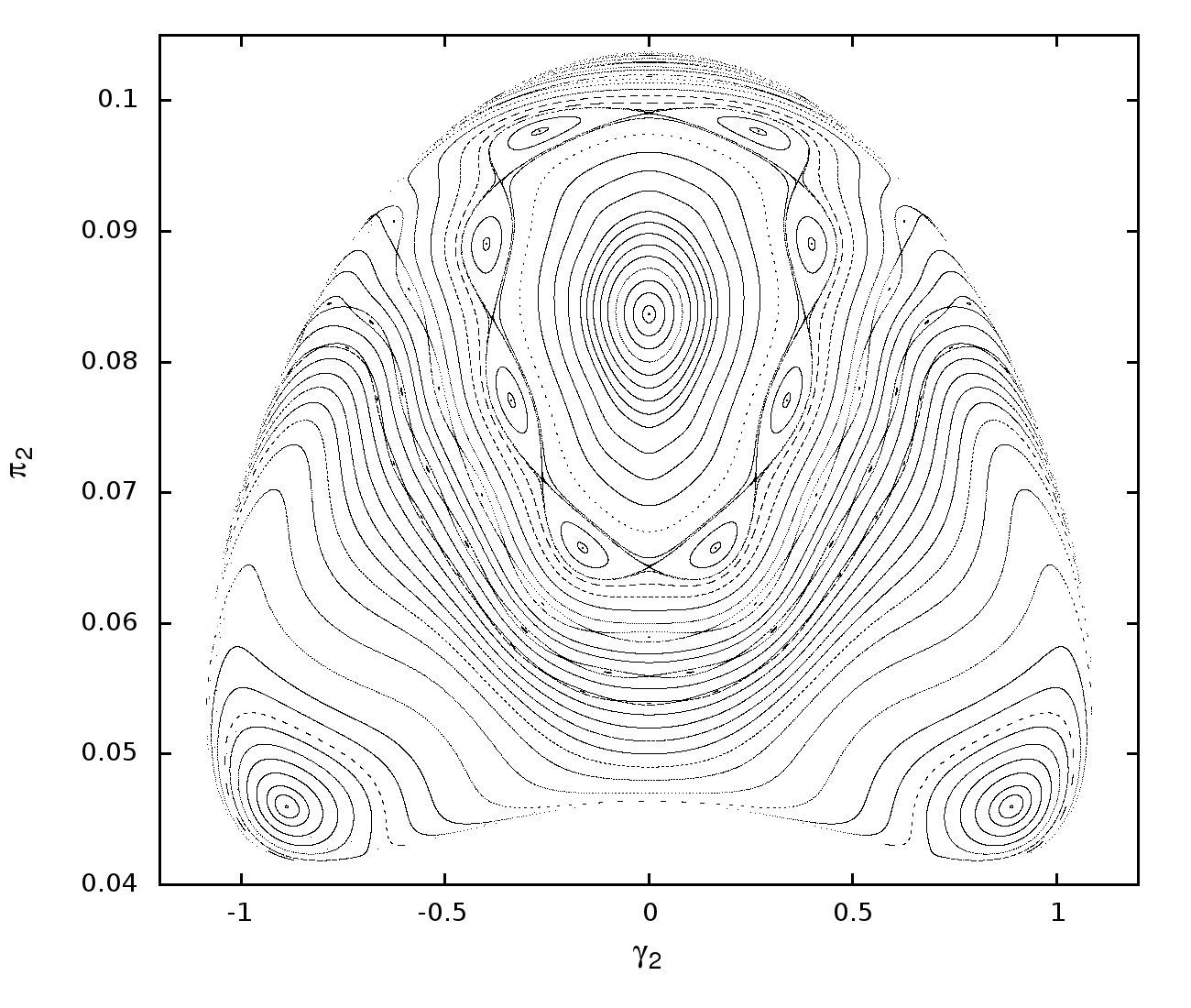}
    \includegraphics[scale=0.16]{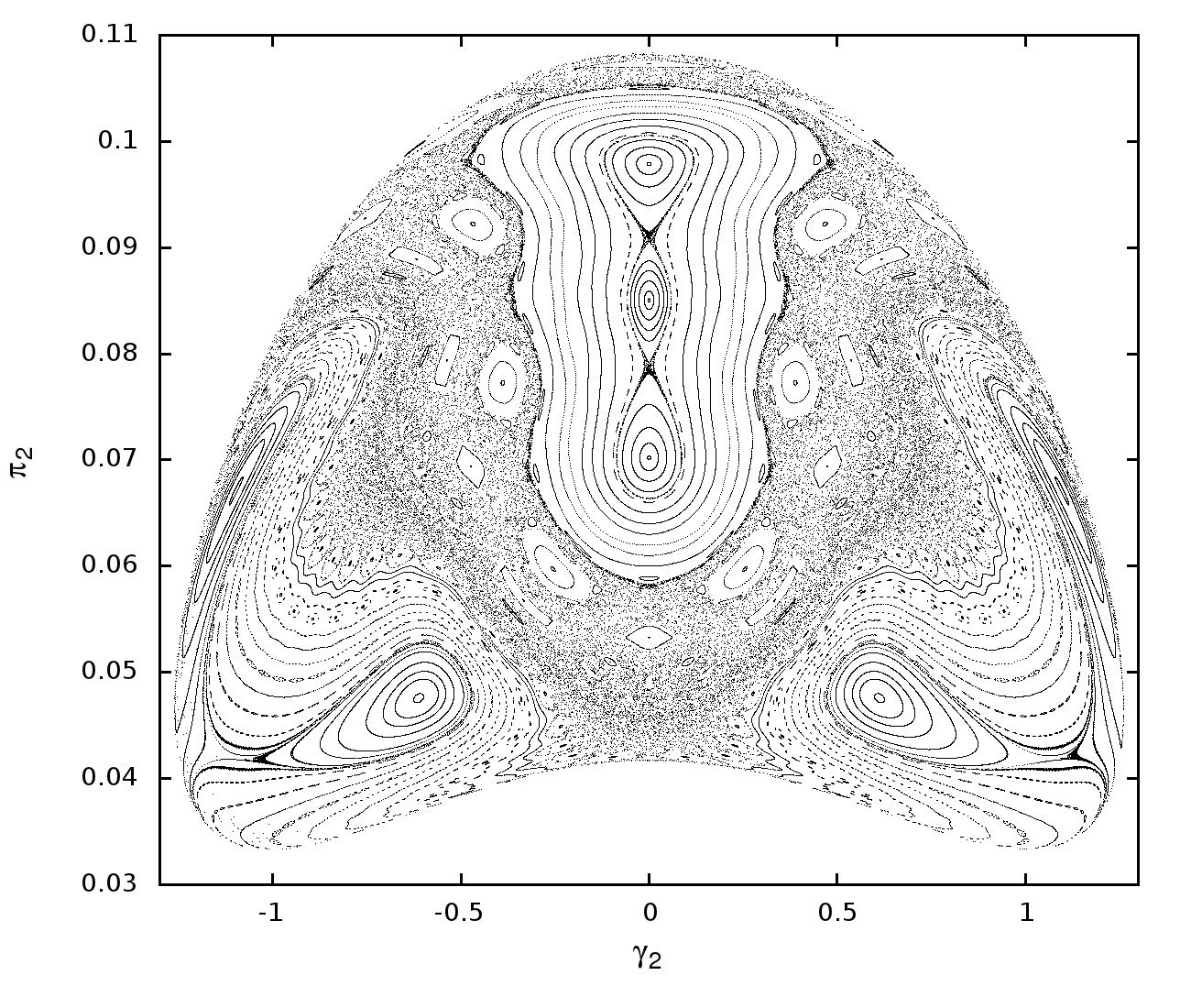}
  \end{center}
  \caption{\small Poincar\'e sections for $ m_1 = 1,\, m_2 = m_3 =
    2,\, l_1 = 2,\, l_2 = l_3 = 1,\, \pi_3 = c = \frac{1}{2}$ and $E =
    0.0035$ (on the left) and $E = 0.0363$ (on the
    right). Cross-section plane $\gamma_1=0$ \label{l33}}
\end{figure}
If  masses and lengths of two pendula attached to the first one are
equal, the  one can expect that behaviour of such system is regular. But
Poincar\'e sections suggest that also in this case the system is not
integrable.

In order to prove that the system is not integrable we apply again the
Morales-Ramis theory.  For $l_3=l_2$ and $m_3=m_2$, the reduced
Hamiltonian~\eqref{eq:tH} and the corresponding equations of
motion~\eqref{eq:systfin} simplify considerably.  In this case, for
$c=0$, the system possesses the following invariant manifold
\[
\cN=\{(x_1,x_2,y_1,y_2)\in\C^4\ |\ x_1=y_1=0 \},
\]
and its restriction to this manifold takes the form
\begin{equation}
  \dot x_1=0,\qquad \dot x_2=-\dfrac{y_2}{l_2^2m_2},\qquad \dot
  y_1=0,\qquad
  \dot y_2=0.
  \label{eq:particular2}
\end{equation} 
These equations determine our particular solution $\vvarphi(t)$.
Variational equations~\eqref{eq:ve} along this particular solution are
also simpler. Namely,  entries $a_{i,4}$ of matrix $\vA$ vanish
for $1\leq i \leq 4$.

Now, equations for $X_1$ and $Y_1$ form a closed subsystem of
variational equations. From them we obtain a second order equation for
$X_1$.  It takes form \eqref{eq:variaty} with coefficients
\[
\begin{split}
  & a=-\dfrac{4 y_2 (l_2 + l_1 \cos x_2) (l_1 (m_1 + 2 m_2) + 2 l_2
    m_2 \cos x_2) \sin x_2}{l_2^2 (2 l_2^2 m_2 + l_1^2 (m_1 + 2 m_2) +
    4 l_1 l_2 m_2 \cos x_2) (m_1 +
    m_2 - m_2 \cos(2x_2))},\\
  &b=\dfrac{2 y_2^2 (l_2 + l_1 \cos x_2) (4 l_1 l_2 m_2 + (2 l_2^2 m_2
    + l_1^2 (m_1 + 2 m_2)) \cos x_2)}{l_1 l_2^4 m_2 (2 l_2^2 m_2 +
    l_1^2 (m_1 + 2 m_2) + 4 l_1 l_2 m_2 \cos x_2) (m_1 + m_2 - m_2
    \cos(2x_2))}.
\end{split}
\]
Changing independent variable according to \eqref{eq:ratio} we obtain
equation of the form~\eqref{eq:zvar} with coefficients
\[
\begin{split}
  p&=\dfrac{-8 z^2 \alpha + 16 z^4 \alpha - 4 \alpha (2 + \gamma) + 6
    z^3 (2 + \alpha^2 (2 + \gamma)) - z (6 + \gamma) (2 + \alpha^2 (2
    + \gamma))}{(z^2-1) (-2 +
    2 z^2 - \gamma) (2 + 4 z \alpha + \alpha^2 (2 + \gamma))},\\
  q&=\dfrac{2 (1 + z \alpha) (4 \alpha + z (2 + \alpha^2 (2 +
    \gamma)))}{(z^2-1) \alpha (-2 + 2 z^2 - \gamma) (2 + 4 z \alpha +
    \alpha^2 (2 + \gamma))}.
\end{split}
\]
The coefficient $r$ in reduced equation \eqref{eq:normal} has the following
expansion
\begin{equation}
  \begin{split}
    &r= -\dfrac{3}{16 (z-1)^2} -\dfrac{3}{16 (z+1)^2}-\dfrac{1}{4 (z
      - \delta)^2} - \dfrac{1}{4 (z + \delta)^2}+\dfrac{3}{4(z-z_5)}\\
    &+\dfrac{8 + \alpha (15 + \delta^2 + \alpha (6 + \delta^2 (2 +
      \alpha (-1 + \delta^2))))}{16 \alpha (-1 + \delta^2) (1 + \alpha
      (2 +
      \alpha\delta^2))(z-1)}\\
    & +\dfrac{8 + 3 \alpha (-5 + 2 \alpha) + \alpha (-1 + \alpha (2 +
      \alpha)) \delta^2 - \alpha^3 \delta^4}{16 \alpha (-1 + \delta^2)
      (1 + \alpha (-2 + \alpha \delta^2))(z+1)}+\dfrac{-2 \delta +
      \alpha (-5 + \delta (\delta + \alpha (-2 + \alpha \delta (-1 +
      \delta^2))))}{4 \alpha
      \delta (1 + \alpha \delta)^2 (-1 + \delta^2) (z - \delta)}\\
    &-\dfrac{2 \delta + \alpha (-5 + \delta (\delta + \alpha (2 +
      \alpha \delta (-1 + \delta^2))))}{4 \alpha \delta (-1 + \alpha
      \delta)^2 (-1 + \delta^2)(z +\delta)} +\dfrac{\alpha (5 + 3
      \alpha^4 \delta^4 + \alpha^6 \delta^6 + \alpha^2 (-16 + 7
      \delta^2))}{(-1 + \alpha^2 \delta^2)^2 (1 + \alpha^4 \delta^4 +
      2 \alpha^2 (-2 + \delta^2))(z-z_5)}
  \end{split}
  \label{eq:rrexy_1}
\end{equation} 
where $z_5=-(1 + \alpha^2 \delta^2)/(2 \alpha)$ and
$\delta^2=1+\gamma/2$.  Equation \eqref{eq:zvar} with coefficient $r$
given in \eqref{eq:rrexy_1} has singularities $z_{1,2}=\pm1$,
$z_{3,4}=\pm\delta$ and $z_5=-(1+\alpha^2\delta^2)/(2\alpha)$, that
are all poles of the second order and infinity has degree 3 for all
positive $\alpha$ and $\delta$.

At first we consider the generic case.
\begin{lemma}
\label{lem:gen}
  If $\delta\neq 1/\alpha$ and $\delta\neq 1$, then the differential
  Galois group of equation \eqref{eq:normal} with $r(z)$ given
  by~\eqref{eq:rrexy_1} is $\mathrm{SL}(2, \C)$.
\end{lemma}
\begin{proof}
  We apply the Kovacic algorithm.  The sets of exponents $E_i$ in this
  case are
  \begin{equation}
    E_1=E_2=\left\{\dfrac14,\dfrac34\right\},\quad
    E_3=E_4=\left\{\dfrac12\right\},\quad
    E_5=\left\{-\dfrac{1}{2},\dfrac{3}{2}\right\},\quad E_6=\{0,1\}.
  \end{equation}
  Hence in a neighbourhood of $z_{\star}=z_3$, and $z_{\star}=z_4$
  local solutions have the form~\eqref{eq:w12}. Thus, we proceed in a
  similar way as in the proofs of previous lemmas.

  There exists only one $e\in E=\prod_{i=1}^{5}E_i$ , such that
  $d(e)\in\N_0$. This is
  \begin{equation}
    e=\left(\dfrac14,\dfrac14,
      \dfrac12,\dfrac12,-\dfrac12,1\right)
  \end{equation}
  and for it we have $d(e)=0$.  Then we construct rational functions
  \[
  \omega=\sum_{i=1}^5\dfrac{e_i}{z-z_i}
  \]
  and check whether monic polynomial $P=1$ satisfies differential
  equation \eqref{eq:P} for some values of parameters $\alpha$ and
  $\delta$. This equation is satisfied only when equalities
  \[
  2 \alpha = 0,\quad 1 + 2 \alpha^2 + \alpha^2 \delta^2 = 0,\quad
  \alpha + \alpha^3 \delta^2 = 0,
  \]
  hold but this system of algebraic equations is not consistent.
\end{proof}

\begin{lemma}
\label{lem:gen1}
  If $\delta= 1/\alpha$ and $\delta\neq 1$, then the differential
  Galois group of equation \eqref{eq:normal} with $r(z)$ given
  by~\eqref{eq:rrexy_1} is $\mathrm{SL}(2, \C)$.
\end{lemma}
\begin{proof}
  In this case $r$ given in~\eqref{eq:rrexy_1} simplifies to the form
  \begin{equation}
    \begin{split}
      r&=-\dfrac{3}{16 (z-1)^2}-\dfrac{3}{16 (z+1)^2}
      -\dfrac{1}{4(z-z_3)^2}-\dfrac{1 +3 \alpha}{16(\alpha-1) (z-1)} +
      \dfrac{-1 + 3 \alpha}{
        16 (\alpha+1)(z+1)}\\
      &+\dfrac{\alpha^3}{2(-1 + \alpha^2) (z-z_3)}-\dfrac{\alpha}{2
        (z-z_4)},
    \end{split}
  \end{equation} 
  where $z_3=\alpha^{-1}$ and $z_4=-\alpha^{-1}$. Singularities
  $z_{1,2}=\pm1$ and $z_3$ are poles of the second order, and $z_4$ is
  a pole of the first order. The infinity $z_5=\infty$ is generically
  of degree 3, i.e., it is a regular singular point. Its degree
  changes only for excluded values of $\alpha$, i.e., for
  $\alpha\in\{0,1\}$.  The difference of exponents at $z_{\star}=z_3$
  is zero.  Thus local solutions around this point have the
  form~\eqref{eq:w12}, and the differential Galois group is either a
  subgroup of the triangular group, or it is $\mathrm{SL}(2,\C)$. Thus we have
to check only the first case of the Kovacic algorithm.

  The sets of exponents $E_i$ are now following
  \[
  E_1=E_2=\left\{\dfrac{1}{4},\dfrac{3}{4}\right\},\qquad
  E_3=\left\{\dfrac{1}{2}\right\},\qquad E_4=E_5=\{0,1\}.
  \]
  Only for $e=(1/4,1/4,1/2,0,1)\in\prod_{i=1}^{5}$, value of $d(e)$ is
  an integer equal to $d(e)=0$. Inserting rational function
  \[
  \omega=\sum_{i=1}^4\dfrac{e_i}{z-z_i},
  \]
  and $P=1$ into differential equation \eqref{eq:P} we obtain $\alpha
  (z + \alpha)=0$.  This condition is fulfilled only for excluded
  value $\alpha=0$.  Thus, the differential Galois group is not a
  subgroup of the triangular group, so it is $\mathrm{SL}(2,\C)$.
\end{proof}
Case when $\delta=1$ corresponds to $\gamma=0$ and now we prove the following
lemma.
\begin{lemma}
\label{lem:gen2}
  For $\gamma=0$ differential Galois group of \eqref{eq:normal} is
  $\mathrm{SL}(n,\C)$.
\end{lemma}
\begin{proof}
  Using particular solution $\vvarphi$ defined by
  \eqref{eq:particular2} one can easy check that the first and the
  second case of the Kovacic cannot occur but an analysis of the third
  case is very onerous.

  Fortunately, in this case we can find another particular solution
  for which analysis of the differential Galois group of the
  corresponding variational equations is simpler.

  Namely, Hamilton equation governed by Hamiltonian~\eqref{eq:hamfin}
  for $m_1=0$, $m_3=m_2$ and $l_3=l_2$ have another particular
  solution lying on invariant manifold
  \[
  \scN=\left\{(\gamma_1,\gamma_2, \pi_1, \pi_2 )\ |\ \gamma_2=0 \quad
    \pi_1=2\pi_2 \right\},
  \]
  We consider only the level $c=0$. It is convenient to make the
  following non-canonical change of variables
  \begin{equation}
    \begin{bmatrix}
      \gamma_1\\
      \gamma_2\\
      \pi_1\\
      \pi_2
    \end{bmatrix}
    = \vA
    \begin{bmatrix}
      x_1\\
      x_2\\
      y_1\\
      y_2
    \end{bmatrix},\qquad \vA=
    \begin{bmatrix}
      0&1&0&0\\
      1&0&0&0\\
      0&0&0&1\\
      0&0&-\dfrac{1}{2}&\dfrac{1}{2}\\
    \end{bmatrix}.
    \label{eq:trafononcan1}
  \end{equation}
  After this transformation equations of motion take the form
  \begin{equation}
    \dot\vx={\mathbb{J}}\nabla_{\vx}\widehat{H},\qquad
    {\mathbb{J}}=\vB\mathbb{I}\vB^T,\quad \vB=\vA^{-1}\qquad  
    \widehat{H}=H(\vA\vx),
    \label{eq:systfin1}
  \end{equation}
  where $\vx=[x_1,x_2,y_1,y_2]^T$.  Hamilton function of the reduced
  system in these variables takes the form
  \[
  \begin{split}
    \widehat H&=[4 l_2^2 y_2^2 + 4 l_1^2 (y_1^2 + y_2^2) -
    l_1 (\cos x_2 (-4 l_2 y_2 (y_1 + y_2) + l_1 (y_1 - y_2)^2 \cos x_2)\\
    & +
    2 (y_1 - y_2) (2 l_2 y_2 + l_1 (y_1 + y_2) \cos x_2) \cos(x_1+x_2)\\
    & + l_1 (y_1 + y_2)^2 \cos(x_1+x_2)^2)]/[4 l_1^2 l_2^2 m_2 (2 -
    \cos(2x_2) - \cos[2(x_1+x_2)])].
  \end{split}
  \]

  In these variables our system has invariant manifold
  \begin{equation}
    \widehat\scN=\left\{(x_1,x_2,y_1,y_2)\ |\ x_1=y_1=0\right\}.
  \end{equation} 
  Equations of motion restricted to this manifold read
  \begin{equation}
    \begin{split}
      \dot x_2&=-\dfrac{y_2 (l_1^2 + l_2^2 + 2 l_1 l_2 \cos x_2)}{2
        l_1^2
        l_2^2 m_2\sin^2x_2},\\
      \dot y_2&=\dfrac{y_2^2 (2 (l_1^2 + l_2^2) \cos x_2 + l_1 l_2 (3
        + \cos(2x_2)))}{4 l_1^2 l_2^2 m_2\sin^3x_2}.
      \label{eq:partki}
    \end{split}
  \end{equation} 
  These equations defines our particular solution and variational
  equations along it take the form
  \begin{equation}
    \begin{bmatrix}
      \dot X_1\\
      \dot X_2\\
      \dot Y_1\\
      \dot Y_2
    \end{bmatrix}
    =
    \begin{bmatrix}
      a_{11}&0&a_{13}&0\\
      a_{21}&a_{22}&a_{23}&a_{24}\\
      a_{31}&0&a_{33}&0\\
      a_{41}&a_{42}&0&a_{44}
    \end{bmatrix}
    \begin{bmatrix}
      X_1\\
      X_2\\
      Y_1\\
      Y_2
    \end{bmatrix},
  \end{equation}
  where
  \[
  \begin{split}
    &a_{11}=-\dfrac{y_2 (l_2 + l_1 \cos x_2)}{2 l_1 l_2^2 m_2\sin
      x_2},\quad
    a_{13}=-\dfrac{1}{l_2^2 m_2)},\\
    &a_{21}=-\dfrac{y_2 ((7 l_1^2 + 8 l_2^2) \cos x_2 + l_1 (10 l_2 +
      6 l_2 \cos(2x_2) + l_1 \cos(3x_2)))}{16 l_1^2 l_2^2 m_2\sin^3
      x_2},\\
    &a_{22}=-\dfrac{y_2 (2 (l_1^2 + l_2^2) \cos x_2 + l_1 l_2 (3 +
      \cos(2x_2)))}{2
      l_1^2 l_2^2 m_2\sin^3 x_2},\quad a_{23}=\dfrac{1}{2 l_2^2 m_2},\\
    &a_{24}=\dfrac{l_1^2 + l_2^2 + 2 l_1 l_2 \cos x_2}{2 l_1^2 l_2^2
      m_2\sin^2
      x_2},\\
    &a_{31}=-\dfrac{y_2^2 (10 l_1 l_2 \cos x_2 + 4 (l_1^2 + 2 l_2^2)
      \cos(2x_2) + l_1 (6 l_2 \cos(3x_2) + l_1 (3 + \cos(4x_2))))}{32
      l_1^2 l_2^2 m_2\sin^4
      x_2},\\
    &a_{33}=\dfrac{y_2 (l_2 + l_1 \cos x_2)}{2 l_1 l_2^2 m_2\sin x_2},\\
    &a_{41}=\dfrac{1}{2}a_{42}=-\dfrac{y_2^2 (23 l_1 l_2 \cos x_2 + 4
      (l_1^2 +
      l_2^2) (2 + \cos(2x_2)) +l_1 l_2 \cos(3x_2))}{16 l_1^2 l_2^2 m_2\sin^4 x_2},\\
    &a_{44}=\dfrac{y_2 (2 (l_1^2 + l_2^2) \cos x_2 + l_1 l_2 (3 +
      \cos(2x_2)))}{2 l_1^2 l_2^2 m_2\sin^3 x_2}
  \end{split}
  \]
  Thus we have the closed subsystem of normal variational equations
  \begin{equation}
    \dot X_1=a_{11}X_1+a_{13}Y_1,\quad \dot Y_1=a_{31}X_1+a_{33}Y_1.
    \label{eq:normki}
  \end{equation}
  This system can be rewritten as one second order linear equation
  \begin{equation}
    \ddot X_1+b(t)X_1=0,\qquad b(t)=-\dfrac{y_2^2 (l_1 + l_2 \cos x_2)^2}{4 l_1^2
      l_2^4 m_2^2\sin^4 x_2}.
    \label{eq:variaty2}
  \end{equation} 
  To rationalise this equation we use transformation
  \begin{equation}
    t\longmapsto z=\cos x_2(t).
    \label{eq:ratio1}
  \end{equation}
  This transformation together with the following expressions on time
  derivatives
  \[
  \begin{split}
    &\Dt=\dot z\Dz,\qquad \DDt=(\dot z)^2\DDz+\ddot z \Dz,\\
    &\dot z=-\sin x_2\,\dot x_2=-\dfrac{y_2 (l_1^2 + l_2^2 + 2 l_1 l_2
      \cos x_2)}{2 l_1^2 l_2^2 m_2\sin x_2},\qquad (\dot
    z)^2=-\dfrac{y_2^2 (l_1^2 + l_2^2 + 2 l_1 l_2 z)^2}{4 l_1^4 l_2^4
      m_2^2
      (-1 + z^2)},\\
    &\ddot z=-\cos x_2\, \dot x_2^2-\sin x_2\ddot x_2=-\dfrac{y_2^2
      (l_1^2 + l_2^2 + 2 l_1 l_2 z)}{4 l_1^3 l_2^3 m_2^2 (-1 + z^2)},
  \end{split}
  \]
  translates equation~\eqref{eq:variaty2} into the following one 
  \begin{equation}
    X''+p(z)X'+q(z)X=0,\qquad X=X_1,
    \label{eq:zvar1}
  \end{equation} 
  with rational coefficients
  \[
  \begin{split}
    & p=\dfrac{\ddot z}{\dot z^2}=\dfrac{1}{2(z-z_3)},\\
    &q=\dfrac{b}{\dot z^2}=\dfrac{1}{4(z-z_3)^2}+\dfrac{\alpha^2}{2 (1
      +\alpha)^2 (z-1)}-\dfrac{\alpha^2}{2 (\alpha-1)^2(
      z+1)}+\dfrac{2 \alpha^3}{(\alpha^2-1)^2(z-z_3)}.
  \end{split}
  \]
  Here $z_3=-(1 + \alpha^2)/(2\alpha)$ and $\alpha=l_1/l_2$. Equation
  \eqref{eq:zvar1} has four singularities
  \[
  z_{1,2}=\pm1,\qquad z_3=-\dfrac{1 + \alpha^2}{2\alpha},\qquad
  z_4=\infty.
  \]
  Using the transformation~\eqref{eq:tran} to \eqref{eq:zvar1} we
  obtain its reduced form
  \begin{equation}
    \label{eq:normal1}
    w'' = r(z) w, \qquad 
    r(z) = -q(z) + \frac{1}{2}p'(z)  + \frac{1}{4}p(z)^2,
  \end{equation}
  where the expansion of coefficient $r$ is the following
  \begin{equation}
    r=-\dfrac{7}{16(z-z_3)}-\dfrac{\alpha^2}{2 (\alpha+1)^2
      (z-1)}+\dfrac{\alpha^2}{2(\alpha-1)^2 (z+1)}-\dfrac{2
      \alpha^3}{(\alpha^2-1)^2(z-z_3)}
  \end{equation} 
  and order of infinity is 2 provided $\alpha\neq 0$.

  At first let us assume that $\alpha\neq 1$.  To check the differential
  Galois group of equation~\eqref{eq:normal1} we apply the Kovacic
  algorithm.  In the first case sets $E_i$ are following
  \begin{equation}
    E_1=E_2=\left\{0,1\right\}, \qquad
    E_3=E_4=\left\{\dfrac{1}{2}-\rmi\dfrac{\sqrt{3}}{4},\dfrac{1}{2}+\rmi\dfrac{
        \sqrt { 3 } }{4}\right\}.
  \end{equation} 
  Next, we select from the Cartesian product $E=\prod_{i=1}^4E_i$
  those elements $e=(e_1,e_2,e_3,e_4)$ for which
  \begin{equation}
    d(e)=e_4-\sum_{i=1}^3e_i\in\N_0.
  \end{equation}
  We have two choices of $e\in E$, namely,
  \begin{equation}
    e^{(1)}=\left(0,0,\dfrac{1}{2}-\rmi\dfrac{\sqrt{3}}{4},\dfrac{1}{2}-
      \rmi\dfrac{\sqrt{3}}{4}\right), \quad
    e^{(2)}=\left(0,0,\dfrac{1}{2}+\rmi\dfrac{\sqrt{3}}{4},\dfrac{1}{2}+
      \rmi\dfrac{\sqrt{3}}{4}\right),
  \end{equation}
  that give $d(e^{(i)})=0$ for $i=1,2$. Then we construct rational
  functions
  \[
  \omega=\sum_{i=1}^3\dfrac{e_i}{z-z_i}
  \]
  and substitute into equation~\eqref{eq:P} setting also $P=1$.  This
  equation is satisfied only for $\alpha=0$ for both choices
  $e^{(i)}$, $i=1,2$.

  In the second case of the Kovacic algorithm sets $E_i$ are the
  following
  \[
  E_1=E_2=\{4\},\qquad
  E_3=E_4=\left\{2-\rmi\sqrt{3},2,2+\rmi\sqrt{3}\right\}\cap\Z=\{2\}.
  \]
  Thus we see that there is no element $e=(e_1,e_2,e_3,e_4)$ for that 
  \begin{equation}
    d(e)=\dfrac{1}{2}\left(e_4-\sum_{i=1}^3e_i\right)\in\N_0.
  \end{equation}
  Similarly in the third case of this algorithm sets $E_i$ are
  following
  \[
  E_1=E_2=\{12\}, \quad E_3=E_4=\left\{6+\dfrac{6k\rmi\sqrt{3}}{n}\ |\
    k=0,\pm1,\pm2,\ldots,\pm\dfrac{n}{2} \right\}\cap\Z=\{6\}
  \]
  where $n=4,6,12$. But we see that also there is no
  $e=(e_1,e_2,e_3,e_4)$ with $e_i\in E_i$ given above for that
  \[
  d(e)=\dfrac{n}{12}\left(e_4-\sum_{i=1}^3e_i\right)\in\N_0.
  \]
  Thus differential Galois group is whole $\mathrm{SL}(2,\C)$.

  For $\alpha=1$ we have a confluence of singular points.  In this
  case coefficient $r$ simplifies to the following forms
  \[
  r=-\dfrac{1 + 7 z}{16 (z-1) (1 + z)^2}.
  \]
  We note that in these cases equation \eqref{eq:normal1} becomes a
  Riemman $P$ equation.  With a help of the Kimura
  theorem~\cite{Kimura:69::} one can show that the differential Galois
  group of this equation is $\mathrm{SL}(2,\C)$.
\end{proof}

The second part of our main Theorem~\ref{thm:main} is an immediate
consequence of Theorem~\ref{thm:MR} and Lemmas~\ref{lem:gen},
\ref{lem:gen1} and \ref{lem:gen2}.
\section{Conclusions}
\label{sec:concl}
The integrability analysis for the flail pendulum with parameters
satisfying
\begin{equation}
  m_2l_2=m_3l_3
  \label{eq:equalit}
\end{equation}
was almost finished except the case when $m_1=0$, $m_2\neq m_3$ and
$l_2\neq l_3$. The only integrable case  was found  for $l_1=0$. 

There is an open question about integrability in the case when
equality \eqref{eq:equalit} does not hold. Poincar\'e sections suggest
non-integrability for generic values of parameters.
\begin{figure}[htp]
  \begin{center}
    \includegraphics[scale=0.165]{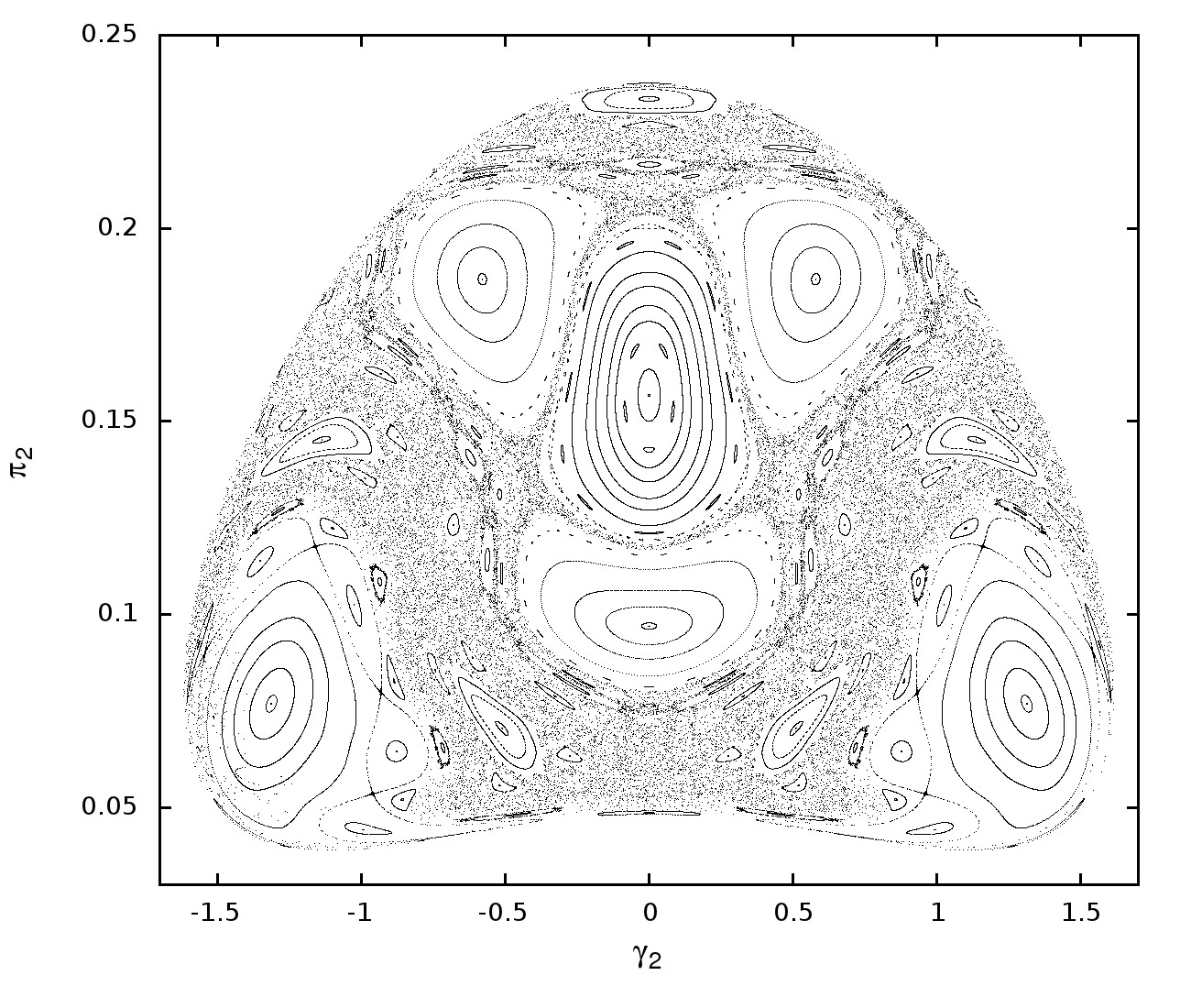}
    \includegraphics[scale=0.165]{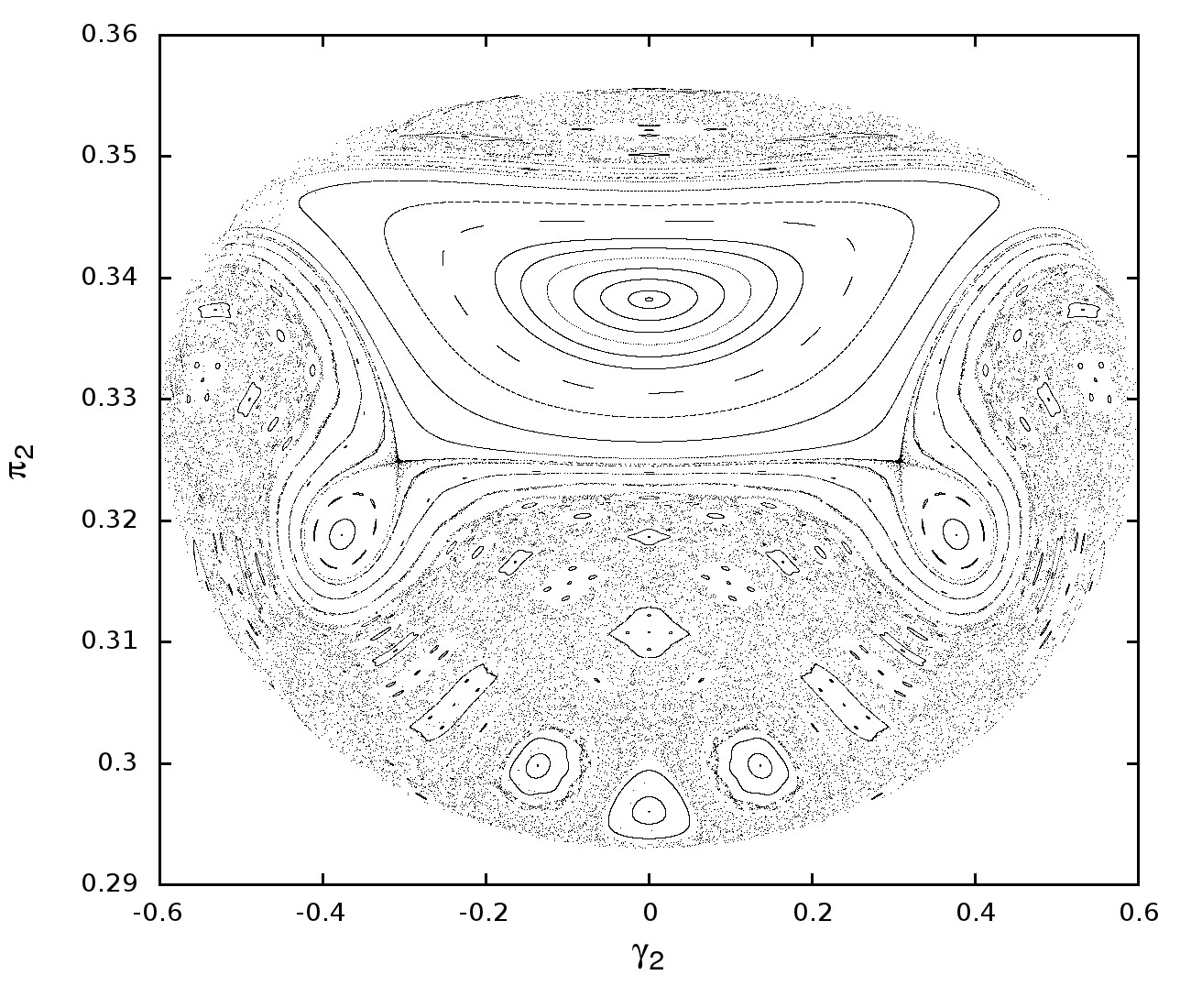}
\caption{\small The Poincar\'e  sections for 
$m_1 =
 m_2 = m_3 = 1,\,  l_1 = l_3=1,\,  l_2 =
2,\, \pi_3 = c = 1$ and  $E = 0.042$  (on the left)  and $ m_1
=
 m_2 =1,\,  m_3 = 2,\,  l_1 = l_2=1,\,  l_3 =
3,\, \pi_3 = c = 1$ and  $E
= 0.00358$  (on the right). Cross-section plane $\gamma_1=0$
\label{l55}}
\end{center}
\end{figure}
\section*{Acknowledgments}
The authors are very grateful to Andrzej J. Maciejewski for many helpful
comments and suggestions concerning improvements and simplifications of some
results.

\appendix
\section{ Explicit forms of reduced equations and variational
equations}
\label{sec:appendix}
The explicit form of vector field $\vv=[v_1,v_2,v_3,v_4]^T$ in
Section~\ref{sec:nierowne} defined by reduced equations of motion
\eqref{eq:systfin} after linear transformation \eqref{eq:trafononcan} is the
following
\begin{equation}
 \begin{split}
 & v_1=\left[2 l_2^2 l_3 m_2 m_3 (c l_3 + l_2 y_2 - l_3 (y_1 + y_2)) - 
   l_1^2 (m_1 + m_2 + m_3) (l_2^2 m_2 y_2 + l_3 m_3 (l_3 y_1 - l_2
y_2))\right.\\
&\left. + 
   l_1 (l_1 l_3 m_3^2 (l_3 y_1 - l_2 y_2) \cos (x_1-x_2)^2 + 
      l_2 m_2 \cos x_2 (l_3 m_3 (l_3 (c - 3 y_1 - y_2) + 3 l_2 y_2)\right.\\
&\left. + 
         l_1 l_2 m_2 y_2 \cos x_2) + 
      l_2 m_2 m_3 \cos (x_1-x_2) (l_2 (l_3 (c - y_1 - 3 y_2) + l_2 y_2) - 
         l_1 (-l_2 y_2\right.\\
&\left. + l_3 (y_1 + y_2)) \cos x_2))\right]/\left[l_1^2
l_2^2 l_3^2 (-m_2 m_3 (m_1 + m_2 + m_3) + 
     m_2 m_3 (m_3 \cos (x_1-x_2)^2 + m_2 \cos x_2^2))\right],\\
&v_2=\left[-l_1^2 (m_1 + m_2 + m_3) (l_3 y_1 - l_2 y_2) + 
   l_2^2 m_2 (c l_3 + l_2 y_2 - l_3 (y_1 + y_2)) + 
   l_1 (l_1 m_3 (l_3 y_1\right.\\
&\left. - l_2 y_2) \cos (x_1-x_2)^2 + 
      l_2 m_2 (c l_3 + 2 l_2 y_2 - l_3 (2 y_1 + y_2)) \cos x_2 - 
      l_2 m_2 y_2 \cos (x_1-x_2)\right.\\
&\left.\times (l_2 + l_1 \cos x_2))\right]/\left[l_1^2 l_2^2
(-l_3 m_2 (m_1 + m_2 + 
        m_3) + l_3 m_2 (m_3 \cos (x_1-x_2)^2 + m_2 \cos x_2^2))\right]\\
&v_3=\left[(l_2 (-l_2 y_2 + l_3 (-c + y_1 + y_2)) + l_1 l_2 y_2 \cos (x_1-x_2)
+ 
     l_1 (l_3 y_1 - l_2 y_2) \cos x_2)\right.\\
&\left.\times (l_1 (l_2 (l_3 m_2 + l_2 (4 m_1 +
m_2 + 4 m_3)) y_2 - 
        2 (l_2 - l_3) m_3 (-l_3 y_1 + l_2 y_2) \cos x_1) \cos x_2 \sin
x_1\right.\\
&\left. - 
     l_1 (l_2 + l_3) (l_2 m_2 y_2 + 2 m_3 (-l_3 y_1 + l_2 y_2) \cos x_1) \cos
(3x_2) \sin x_1 + 2 l_2^2 m_3 (-l_2 y_2 + l_3 (-c + y_1\right.\\
&\left. + y_2)) \cos
(2x_2) \sin (2x_1) + 
     l_1 (l_2 (l_3 m_2 - l_2 (4 m_1 + 3 m_2 + 4 m_3)) y_2 \cos x_1 + (l_3 y_1 - 
           l_2 y_2)\right.\\
&\left.\times (2 l_3 (2 (m_1 + m_2) + m_3) + (-l_2 + l_3) m_3
\cos (2x_1))) \sin x_2 + 2 l_2 (c l_3 + l_2 y_2 - l_3 (y_1 + y_2))\right.\\
&\left.\times
(-l_3 m_2 + l_2 m_3 \cos (2x_1)) \sin (2x_2) + 
     l_1 (l_2 + l_3) (l_2 m_2 y_2 \cos x_1 + m_3 (-l_3 y_1 + l_2 y_2) \cos
(2x_1)) \right.\\
&\left.
\times\sin (3x_2))\right]/\left[l_1^2 l_2^2 l_3^3 (2 m_1 + m_2 + m_3 - m_3
\cos[2 (x_1 - x_2)] - m_2 \cos (2x_2))^2\right]\\
&v_4=\left[2 (l_2 (-l_2 y_2 + l_3 (-c + y_1 + y_2)) + l_1 l_2 y_2 \cos (x_1-x_2)
+ 
     l_1 (l_3 y_1 - l_2 y_2) \cos x_2)\right.\\
&\left.\times (2 m_3 \cos (x_1-x_2) (l_2 (-l_2
y_2 + l_3 (-c + y_1 + y_2)) +l_1 (l_3 y_1 - l_2 y_2) \cos x_2) + l_1 l_2 y_2 (2
m_1\right.\\
&\left.+ m_2 + 2 m_3 - m_2 \cos (2x_2))) \sin (x_1-x_2)\right]/\left[l_1^2
l_2^2 l_3^2 (2 m_1 + m_2 + m_3 - m_3 \cos[2 (x_1 - x_2)]\right.\\
&\left. -  m_2 \cos
(2x_2))^2\right]
 \end{split}
\label{eq:ukk}
\end{equation} 

Entries of matrix in variational equations~\eqref{eq:wariaty1} are the following
\begin{equation}
\begin{split}
&a_{11}=\dfrac{c (2 l_2 l_3 + l_1 (l_2 + l_3) \cos x_2) \sin x_2}{l_1 l_2 (l_2 -
l_3) (-l_2 m_2 -l_3 (m_1 + m_2) + (l_2 + l_3) m_2 \cos x_2^2)},\\
&a_{13}=-\dfrac{2 l_2^2 l_3 m_2 + l_1^2 (l_3 m_1 + (l_2 + l_3) m_2) + 
 l_1 l_2 (l_2 + 3 l_3) m_2 \cos x_2}{l_1^2 l_2^2 m_2 (-l_2 m_2 - 
   l_3 (m_1 + m_2) + (l_2 + l_3) m_2 \cos x_2^2)},\\
&a_{14}=\dfrac{(l_2 - l_3) (2 l_2 l_3 + l_1 (l_2 + l_3) \cos x_2)}{l_1^2 l_2 l_3
(-l_2 m_2 - l_3 (m_1 + m_2) + (l_2 + l_3) m_2 \cos x_2^2)},\\
&a_{21}=\dfrac{c l_3 (l_2 + l_1 \cos x_2) \sin x_2}{l_1 l_2 (l_2 - l_3) (-l_2
m_2 - 
   l_3 (m_1 + m_2) + (l_2 + l_3) m_2 \cos x_2^2)},\\
&a_{23}=\dfrac{-l_2^2 l_3 m_2 - l_1^2 (l_2 m_2 + l_3 (m_1 + m_2)) + 
 l_1 l_2 m_2 \cos x_2 (-2 l_3 + l_1 \cos x_2)}{l_1^2 l_2^2 m_2 (-l_2 m_2 - 
   l_3 (m_1 + m_2) + (l_2 + l_3) m_2 \cos x_2^2)},\\
&a_{24}=-\dfrac{1}{l_2 l_3 m_2}+\dfrac{(l_3 - l_2) (l_2 + l_1 \cos x_2)}{l_1^2
l_2 (l_2 m_2 +l_3 (m_1 + m_2) - (l_2 + l_3) m_2 \cos x_2^2)},\\
&a_{31}=\dfrac{2 c^2 (l_2 + l_3) \sin x_2^2}{(l_2 -l_3)^2 (-2 l_3 m_1 - (l_2 +
l_3) m_2 + (l_2 + l_3) m_2 \cos(2x_2))},\\
&a_{33}=-\dfrac{2 c (l_2 + l_3) (l_2 + l_1 \cos x_2) \sin x_2}{l_1 l_2 (l_2 -
l_3) (-2 l_3 m_1 - (l_2 + l_3) m_2 + (l_2 + l_3) m_2 \cos(2x_2))},\\
&a_{34}=\dfrac{2 c (l_2 + l_3) \sin x_2}{l_1 l_3 (-2 l_3 m_1 - (l_2 + l_3) m_2
+ (l_2 + l_3) m_2 \cos(2x_2))},\\
&a_{41}=\dfrac{2 c^2 l_3 \sin x_2^2}{(l_2 - 
   l_3)^2 (-2 l_3 m_1 - (l_2 + l_3) m_2 + (l_2 + l_3) m_2 \cos(2x_2))},\\
&a_{43}=-\dfrac{2 c l_3 (l_2 + l_1 \cos x_2) \sin x_2}{l_1 l_2 (l_2 - 
   l_3) (-2 l_3 m_1 - (l_2 + l_3) m_2 + (l_2 + l_3) m_2 \cos(2x_2))},\\
&a_{44}=\dfrac{2 c \sin x_2}{l_1 (-2 l_3 m_1 - (l_2 + l_3) m_2 + (l_2 + l_3) m_2
\cos(2x_2))}.
\end{split}
\label{eq:wariaty1_entries}
\end{equation}

\section{Linear second order differential equation with rational
coefficients and Kovacic algorithm}
\label{app:kov}
Let us consider a linear second order differential equation with rational
coefficients
\begin{equation}
\label{eq:agl}
w'' + p(z) w' + q(z) w =0, \qquad p(z), q(z)\in \C(z).
\end{equation}
A point $z=c\in\C$ is a singular point of this equation if it is a
pole of $p(z)$ or $q(z)$. A singular point is a \emph{regular
  singular} point if at this point $\tilde p(z)=(z-c)p(z)$ and $\tilde
q(z)=(z-c)^2q(z)$ are holomorphic. An \emph{exponent} of
equation~\eqref{eq:agl} at point $z=c$ is a solution of the
\emph{indicial equation}
\[
   \rho(\rho-1) +p_0\rho + q_0 = 0, \qquad p_0=\tilde p(c), \quad
  q_0=\tilde q(c).
\]
After a change of the dependent variable $z\rightarrow 1/z$
equation~\eqref{eq:agl} reads
\begin{equation}
\label{eq:infagl}
\begin{split}
&w'' + P(z) w' + Q(z) w =0, \\
& P(z)=-\frac{1}{z^2}p\left(\frac{1}{z}\right) +\frac{2}{z}, \qquad
Q(z)= \frac{1}{z^4}q\left(\frac{1}{z}\right) .
\end{split}
\end{equation}
We say that the point $z=\infty$ is a singular point for
equation~\eqref{eq:agl} if $z=0$ is a singular point
of~equation~\eqref{eq:infagl}. Equation~\eqref{eq:agl} is called
\emph{Fuchsian} if all its singular points (including infinity) are
regular, see \cite{Whittaker:35::,Ince:44::}.

The Kovacic algorithm \cite{Kovacic:86::} allows to decide whether all
solutions of  equation~\eqref{eq:agl} with rational coefficients
$p(z)$ and $q(z)$ are Liouvillian.  Roughly speaking, Liouvillian
functions are obtained from the rational functions by a finite sequence of
admissible operations:
 solving algebraic equations, integration and taking exponents of
integrals. For a formal definition, see e.g., \cite{Kovacic:86::}. 

If one (non-zero) solution $w_1$ of equation~\eqref{eq:agl} is
Liouvillian, then all its solutions are Liouvillian.  In fact, the
second solution $w_2$, linearly independent from $w_1$, is given by
\[
w_2 = w_1\int \frac{1}{w_1^2}\exp\left[-\int p\right].
\]
Putting
\[
   w = y \exp\left[-\frac{1}{2}\int p \right]
\]
into equation~\eqref{eq:agl}, we obtain its \emph{reduced form}
\begin{equation}
\label{eq:gso}
 y''=r(z) y, \qquad r(z) = -q(z) + \frac{1}{2}p'(z)  + \frac{1}{4}p(z)^2.
\end{equation}
This change of variable does not affect the Liouvillian nature of the
solutions.  For equation~\eqref{eq:gso}, its differential Galois group
$\cG$ is an algebraic subgroup of $\mathrm{SL}(2,\C)$. The following
lemma describes all possible types of $\cG$ and relates these types to the
forms of a solution of \eqref{eq:gso}, see
\cite{Kovacic:86::,Morales:99::c}.
\begin{lemma}
\label{lem:alg}
Let $\cG$ be the differential Galois group of equation~\eqref{eq:gso}.
Then one of four cases can occur.
\begin{description}
\item[Case I] $\cG$ is conjugate to a subgroup of the triangular group
\[
\cT = \left\{ \begin{bmatrix} a & b\\
                                0 & a^{-1}
                      \end{bmatrix}  \; \biggl| \; a\in\C^*, b\in\C\right\} ;
\]
in this case equation \eqref{eq:gso} has an exponential solution of
the form $y=P\exp\int\omega$, where $P\in\C[z]$ and $\omega\in\C(z)$,
\item[Case II] $\cG$ is conjugate to a subgroup of
\[
\cD^\dag = \left\{ \begin{bmatrix} c & 0\\
                                0 & c^{-1}
                      \end{bmatrix}  \; \biggl| \; c\in\C^*\right\} \cup
                      \left\{ \begin{bmatrix} 0 & c\\
                                -c^{-1} & 0
                      \end{bmatrix}  \; \biggl| \; c\in\C^*\right\};
\]
  in this case equation
  \eqref{eq:gso} has a  solution of the form $y=\exp\int \omega$, where
  $\omega$ is algebraic over $\C(z)$ of degree 2,
\item[Case III] $\cG$ is primitive and finite; in this case all
  solutions of equation~\eqref{eq:gso} are algebraic, thus
  $y=\exp\int\omega$, where $\omega$ belongs to an algebraic extension
  of ${\C}(z)$ of degree $n=4,6$ or 12.

\item[Case IV] $\cG= \mathrm{SL}(2,\C)$ and equation \eqref{eq:gso}
  has no Liouvillian solution.
\end{description}
\end{lemma}

Kovacic in paper \cite{Kovacic:86::} have formulated the necessary conditions
for the respective cases from Lemma~\ref{lem:alg} to hold.

At first we introduce notation.
We write $r(z)\in\mathbb{C}(z)$ in the form
\begin{equation*}
r(z) = \frac{s(z)}{t(z)}, \qquad s(z),\, t(z) \in \mathbb{C}[z],
\end{equation*}
where $s(z)$ and $t(z)$ are relatively prime polynomials and $t(z)$ is
monic.  The roots of $t(z)$ are poles of $r(z)$. We denote  $\Sigma':=
\{ c\in\mathbb{C}\,\vert\, t(c) =0 \}$ and  $\Sigma:=\Sigma'\cup\{\infty\}$.
 The
order $\mathrm{ord}(c)$ of $c\in\Sigma'$ is equal to the multiplicity of $c$
as a root of $t(z)$, the order of infinity is defined by
\[
\mathrm{ord}(\infty):=\deg t - \deg s.
\]

\begin{lemma}
\label{lem:neces}
The  necessary conditions for the respective cases
  in  Lemma~\ref{lem:alg} are the following.
  \begin{description}
\item[Case~I.] Every pole of $r$ must have even order or else have order~1.
 The  order of $r$ at $\infty$ must be even or else be greater than~2.
\item[Case~II.] $r$ must have at least one pole that either has odd order
greater than~2 or else has order~2.
\item[Case~III.] The order of a pole  of $r$ cannot
  exceed~2 and the order of $r$ at $\infty$ must be at least 2. If the partial
fraction expansion of $r$ is
\begin{equation}
\label{eq:r}
r(z)=\sum_{i}\dfrac{a_i}{(z-c_i)^2}+
\sum_{j}\dfrac{b_j}{z-d_j},
\end{equation}
then  $\Delta_i=\sqrt{1+4a_i}\in\mathbb{Q}$, for each $i$, $\sum_{j}b_j=0$
and if
\begin{equation*}
g=\sum_ia_i+\sum_jb_jd_j,
\end{equation*}
then $\sqrt{1+4g}\in\mathbb{Q}$.
\end{description}
\label{kovacic}
\end{lemma}
In \cite{Kovacic:86::} Kovacic also formulated a procedure, called now
the Kovacic algorithm, which allows to decide if an equation of the
form~\eqref{eq:gso} possesses a Liouvillian solution and to find it in a
constructive way.  Applying it, we also obtain information about the
differential Galois group of the analysed equation.  Beside the original
formulation of this algorithm several its versions and
improvements are known \cite{Duval:92::,Maciejewski:01::,Morales:99::c}.

Now we describe the Kovacic algorithm for the respective cases
 from Lemma~\ref{lem:alg}.\\[\bigskipamount]
\noindent
\textsc{Case I}

\noindent
\textbf{Step I.} 
Let $\Gamma$ is the set of poles of $r$. For each $c\in\Gamma\cup\{\infty\}$ we
define a rational function $[\sqrt{r}]_c$ and two complex numbers
$\alpha_c^{+}$, $\alpha_c^{-}$ as described below.
\begin{description}
\item[(c$_1$)] If $c\in\Gamma$ and $\ord(c)=1,$ then 
\[
[\sqrt{r}]_c=0,\qquad \alpha_c^{+}=\alpha_c^{-}=1.
\]
\item[(c$_2$)] If $c\in\Gamma$ and $\ord(c)=2,$ then
\[
[\sqrt{r}]_c=0.
\]
If $r$ has the  expansion of the form \eqref{eq:r}, then
\begin{equation}
\alpha_c^{\pm}=  \frac{1}{2}\left(1\pm\sqrt{1+4a}\right)
\label{eq:del1s}
\end{equation}
\item[(c$_3$)]
If $c\in\Gamma$ and $\ord{c}=2\nu\geq4 $ with $k\ge 2$ (only even orders are
admissible in
this case), then
\begin{align*}
  \left[\sqrt{r}\right]_c&= \frac{a}{(z-c)^{\nu}} + \cdots+\frac{d}{(z-c)^2},\\
  r - \left[\sqrt{r}\right]_c^2 &= \frac{b}{(z-c)^{\nu+1}} +
  O\left(\frac{1}{(z-c)^\nu}\right),
\end{align*}
and
\[
\alpha_c^{\pm}= \frac{1}{2}\left(
 \pm\frac{b}{a}+\nu\right).
\]
\item[($\infty_1$)] If the order of $r$ at infinity is $>2$, then
\[
[\sqrt{r}]_{\infty}=0,\qquad \alpha_{\infty}^{+}=0,\quad \alpha_{\infty}^{-}=1.
\]
\item[($\infty_2$)] If the order of $r$ at infinity is 2, then
\[
[\sqrt{r}]_{\infty}=0.
\]
If the Laurent series expansion of $r$ at $\infty$ takes the form
\begin{equation}
\label{eq:Linfty}
r =  \frac{a}{z^2} +
 O\left(\frac{1}{z^3}\right),
\end{equation}
then
\begin{equation}
\alpha_{\infty}^{\pm}=  \frac{1}{2}\left(1\pm\sqrt{1+4a}\right).
\label{eq:delinf}
\end{equation}
\item[($\infty_3$)] If the order of $r$ at $\infty$ is $-2\nu\leq0$ (necessarily
even in this case), then 
\begin{equation}
 \sqrt{r}_{\infty}=az^{\nu}+\cdots+s,\qquad 
\label{eq:inftyexp}
\end{equation} 
is the indicated part of the Laurent series expansion of $\sqrt{r}$ at $\infty$
and
\begin{equation*}
  r - \left[\sqrt{r}\right]_\infty^2 = b z^{\nu-1} +
  O\left(z^{\nu-2}\right).
\end{equation*}
Then 
\begin{equation}
\label{eq:ainf}
\alpha_{\infty}^{\pm}= \frac{1}{2}\left(
 \pm\frac{b}{a}-\nu\right).
\end{equation}
\end{description}

\textbf{Step II.} For each family $s=(s(c),s(\infty))$, $c\in\Gamma$, where
$s(c)$ and $s(\infty)$ are either $+$ either $-$, let
we compute
\[
  d := e_\infty^{s(\infty)}- \sum_{c\in\Gamma}\alpha_c^{s(c)}.
\]
If $d$ is a non-negative integer, then
\begin{equation}
 \omega(z) =
\sum_{c\in\Gamma}\left(\mathrm{s}(c)\left[\sqrt{r}\right]_c+\frac{\alpha_c^{
s(c)}}{z-c}
\right) +
\mathrm{s}(\infty)\left[\sqrt{r}\right]_{\infty},
\label{eq:t1}
\end{equation}
is a candidate for $\omega$.  If there are no such elements,
equation~\eqref{eq:gso} does
not have an exponential solution and the algorithm stops here.

\noindent
\textbf{Step III.} For each family from step II that gives 
$d\in\N_0$ we  search for a monic polynomial $P=P(z)$
of degree $d$ satisfying the
following equation
\begin{equation}
\label{eq:P1}
P'' + 2\omega(z)P' +( \omega'(z)+ \omega(z)^2 -r(z))P =0.
\end{equation}
If such polynomial exists, then equation~\eqref{eq:gso} possesses an
exponential solution of the form $y=P\exp\int\theta$, where $\theta=\omega$,
if not,
equation~\eqref{eq:gso} does not have an exponential solution.
\\[\bigskipamount]
\noindent
\textsc{Case II}

Let $\Gamma$ be the set of poles of $r$.

\noindent
\textbf{Step I.}
For each $c\in\Gamma$ we define $E_c$ as follows
\begin{description}
\item[(c$_1$)] If $\ord(c)=1,$ then the set $E_{c}=\{4\}$ i
\item[(c$_2$)] If $\ord{c}=2$ and $r$ has the  expansion of the form
\eqref{eq:r},
then
\begin{equation}
  E_c := \left\{ 2 , 2\left(1+\sqrt{1+4a}\right),
   2\left(1-\sqrt{1+4a}\right)\right\}\cap\Z,
\label{eq:del2s}
\end{equation}
\item[(c$_3$)] If $\ord{c}=\nu>2$, then $E_{c}=\{\nu\}$
\item[($\infty_{1}$)] If $\ord{\infty}>2$, then $E_{\infty}=\{1,2,4\}$
\item[($\infty_{2}$)] If $\ord{\infty}=2$ and the he Laurent series expansion of
$r$ at $\infty$ takes the form \eqref{eq:Linfty}, then
\begin{equation}
  E_\infty := \left\{2,  2\left(1+\sqrt{1+4a}\right),
  2\left(1-\sqrt{1+4a}\right)\right\}\cap\Z.
\label{eq:del2i}
\end{equation}
\item[($\infty_{3}$)] If $\ord{\infty}=\nu<2$, then $E_{\infty}=\{\nu\}$
\end{description}

\noindent
\textbf{Step II.} We consider all families $(e_c)_{c\in\Gamma\cup\{\infty\}}$
with $e_c\in E_c$ and at least one of the coordinates is odd. 
We compute
\[
  d := \frac{1}{2}\left(e_\infty- \sum_{c\in\Gamma}e_c\right).
\]
and  select those  families $(e_c)_{c\in\Gamma\cup\{\infty\}}$ for which $d(e)$
is a non-negative
integer.  If there are no such elements,  Case II cannot occur and the
algorithm stops here.

\noindent
\textbf{Step III.} For each family  giving
$d\in\N_0$ we define
\begin{equation}
\omega=\omega(z) = \frac{1}{2}
\sum_{c\in\Gamma}\frac{e_c}{z-c},
\label{eq:t2}
\end{equation}
and we search for a monic polynomial $P=P(z)$ of degree $d$ satisfying the
following equation
\begin{equation}
P''' + 3\omega P'' +(3 \omega^2 + 3\omega' -4r)P' +
(\omega'' + 3 \omega\omega' + \omega^3 -4 r\omega - 2r')P =0.
\label{eq:P2}
\end{equation}
If such a polynomial exists, then equation~\eqref{eq:gso} possesses a
 solution of the form $y=\exp\int\theta$, where
\[
\theta^2 - \psi\theta +\frac{1}{2}\psi' + \frac{1}{2}\psi^2 - r =0, \qquad
\psi = \omega + \frac{P'}{P}.
\]
If we do not find such  polynomial, then Case II in Lemma~\ref{lem:alg}
cannot occur.

\noindent
\textsc{Case III}

\noindent
\textbf{Step I.} For each $c\in\Gamma\cup\{\infty\}$ we define a set $E_c$ of
integers as follows
\begin{description}
\item[($c_1$)] If $c$ is a pole of order 1, then $E_2=\{12\}$.

\item[($c_2$)] If $\ord c=2$ and $r$ has the  expansion of the form
\eqref{eq:r},
then
\begin{equation}
E_c=\left\{6+\dfrac{12k}{m}\sqrt{1+4a}\,|\, k=0,\pm1,\pm 2\ldots,\pm
\frac{m}{2}\right\}\cap\Z.
\label{eq:del3s}
\end{equation}
Here and below in this case $m\in\{4,6,12\}$.
%
%
\item[($\infty$)] For $c=\infty$ when the Laurent expansion of $r$
at infinity is of the form \eqref{eq:Linfty} we define independently on order 
\begin{equation}
E_{\infty}=\left\{6+\dfrac{12k}{m}\sqrt{1+4a}\,|\, k=0,\pm1,\pm 2\ldots,\pm
\frac{m}{2}\right\}\cap\Z.
\label{eq:del3i}
\end{equation}
\normalmarginpar
Obviously it can appear that $a_{\infty}=0$.
\end{description}
\noindent
\textbf{Step II.} For $e\in E$ we calculate
\[
  d(e) := \frac{m}{12}\left(e_\infty- \sum_{c\in\Gamma}e_c\right).
\]
We select those elements $e\in E$ for which $d(e)$ is a non-negative
integer.  If there are no such elements,  Case II cannot occur and the
algorithm stops here.

\noindent
\textbf{Step III.} We consider  all families $(e_c)_{c\in\Gamma\cup\{\infty\}}$
with $e_c\in E_c$ that give
$d(e)=n\in\N_0$ . For each such family we define
\begin{equation*}
\omega=\omega(z) = \frac{m}{12}
\sum_{c\in\Gamma}\frac{e_c}{z-c}.
\end{equation*}
Next we search for a monic polynomial $P=P(z)$ of degree $n$
satisfying a differential equation of degree $m+1$ defined in the
following way. Put $P_m =P$. Then calculate $P_i$ for
$i=m,m-1,\ldots,0$, according to the following formula 
\begin{equation}
P_{i-1}=-SP_i'+[(m-i)S'-S\omega]P_i-(m-i)(i+1)S^2rP_{i+1},
\label{eq:P3}
\end{equation}
where
\[
S=\prod_{c\in\Sigma'}(x-c).
\]
Then $P_{-1}=0$ gives the desired equation for $P$. 
If such  polynomial exists, then equation~\eqref{eq:gso} possesses a
 solution of the form $y=\exp\int\theta$, where $\theta$ is a solution of
 the equation
\[
\sum_{i=0}^n\frac{S^iP_i}{(m-i)!}\theta^i=0.
\]
If we do not find such polynomial, then we repeat these calculations
for the next element $m\in\{4,6,12\}$. If for all $m$ such polynomial does
not exist, then Case III in Lemma~\ref{lem:alg} cannot occur.
\newcommand{\noopsort}[1]{}\def\polhk#1{\setbox0=\hbox{#1}{\ooalign{\hidewidth
  \lower1.5ex\hbox{`}\hidewidth\crcr\unhbox0}}} \def\cprime{$'$}
  \def\cydot{\leavevmode\raise.4ex\hbox{.}} \def\cprime{$'$} \def\cprime{$'$}
  \def\cprime{$'$} \def\cprime{$'$}

\end{document}